\documentclass{article}

\usepackage{arxiv}

\usepackage[utf8]{inputenc} 
\usepackage[T1]{fontenc}    
\usepackage{booktabs}       
\usepackage{amsmath}
\usepackage{amsfonts}       
\usepackage{nicefrac}       
\usepackage{microtype}      
\usepackage{doi}
\usepackage{natbib}
\usepackage[english]{babel}
\usepackage{enumitem}
\usepackage{graphicx}
\usepackage{subcaption}
\usepackage{wasysym}
\usepackage{multimedia}
\usepackage{amsfonts}
\usepackage{amssymb}
\usepackage{amsthm}
\usepackage{libertine}
\usepackage{lmodern}
\usepackage{fancyhdr}
\usepackage{color}
\usepackage{dsfont}
\usepackage{bbm}
\usepackage{multicol}
\usepackage{bm} 
\usepackage{ulem}
\usepackage{tabularx} 
\usepackage{mathrsfs}
\usepackage{stmaryrd}
\usepackage{array,multirow,makecell}
\setcellgapes{1pt}
\makegapedcells%
\usepackage{systeme}
\usepackage{comment}
\usepackage{tcolorbox}
\usepackage{cases}
\usepackage{algorithm}
\usepackage{algorithmic}
\usepackage{mathtools}
\usepackage{colortbl}
\usepackage{hhline}

\usepackage[toc,page,title,titletoc,header]{appendix} 

\newtheorem{prop}{Proposition}

\newtheorem{lemma}{Lemma}

\usepackage{afterpage}

\title{Information criteria exploiting latent structure for model selection in Structural Equation Models}

\date{} 					

\author{ Marion Naveau \\ L'Institut Agro, Université de Rennes, CNRS, UMR 6625 IRMAR, Rennes, France. \\
    \And Magalie Houée-Bigot \\ L'Institut Agro, Rennes, France. \\
    \And Matthieu Marbac \\ Université Bretagne Sud, UMR CNRS 6205, LMBA, F-56000 Vannes, France. \\
    \And 
	Anouk Zancarini \\ IGEPP, INRAE, Institut Agro, Univ Rennes, 35653, Le Rheu, France. \\
	\And 
	Mathieu Emily \\ L'Institut Agro, Université de Rennes, CNRS, UMR 6625 IRMAR, Rennes, France. \\
}

\renewcommand{\headeright}{} 
\renewcommand{\undertitle}{} 
\renewcommand{\shorttitle}{Model selection in Structural Equation Models (SEM)}

\hypersetup{
pdftitle={Model selection in Structural Equation Models (SEM)},
pdfsubject={},
pdfauthor={Mathieu ~Emily, Magalie ~Houée-Bigot, Matthieu ~Marbac, Marion ~Naveau, Anouk ~Zancarini},
pdfkeywords={Structural Equation Models,  Model selection, ICL criterion,  BIC criterion, Importance Sampling},
}

\def\keywordname{{\bfseries Keywords}}%
\def\keywords#1{\par\addvspace\medskipamount{\rightskip=0pt plus1cm
\def\and{\ifhmode\unskip\nobreak\fi\ $\cdot$
}\noindent\keywordname\enspace\ignorespaces#1\par}}

\begin{document}
\maketitle

\begin{abstract}
Structural equation models (SEM) are widely used to describe dependency structures between latent variables, making model selection a key issue in many applications. Existing information criteria are generally based on the integrated observed-data likelihood and therefore do not explicitly account for the latent structure of the model. In this paper, we propose two new information criteria derived from the integrated complete-data likelihood. The first adapts the Integrated Completed Likelihood criterion to Gaussian SEM, while the second proposes an alternative approach to approximating the integrated observed-data log-likelihood by incorporating latent structural information and using an importance sampling strategy. Their performance is assessed through an extensive simulation study covering null, direct, indirect and complete latent structures under different sample sizes and signal strengths. The results show that the proposed importance sampling strategy provides robust and competitive model selection across a wide range of scenarios, whereas the proposed ICL criterion is particularly effective for recovering latent dependency structures when the latent variables are accurately estimated. These findings demonstrate the potential benefits of explicitly exploiting the latent structure when developing information criteria for structural equation models.

\end{abstract}


\section{Introduction}

Structural Equation Models (SEM), first introduced by \cite{joreskog1970general}, provide a flexible statistical framework for modelling complex dependency structures involving latent variables. 
They are widely used in several scientific fields such as psychology, social sciences, ecology and econometrics, where some underlying mechanisms cannot be directly observed but only indirectly measured through manifest variables. SEM combine two complementary components: a measurement model linking observed variables to latent variables, and a structural model describing the dependency relationships between the latent variables themselves. 

A central issue in SEM is model selection \citep{bollen2014bic, lin_selecting_2017, preacher_model_2023}. In practice, the dependency structure between latent variables is rarely known in advance and must therefore be inferred from the data. This problem consists in determining which relationships should be included in the latent graph structure. Classical approaches generally rely on information criteria such as the Akaike's Information criterion (AIC, \cite{akaike1974new}) and the Bayesian Information Criterion (BIC, \cite{schwarz1978estimating}). The AIC and BIC criteria evaluate models by balancing fit and parsimony. Both include a negative log-likelihood term, which assesses a model’s ability to reproduce the observed data, as well as a complexity penalty that increases with the number of parameters. Models with lower AIC or BIC values are generally preferred, as they offer a better balance between explanatory power and simplicity. Additional discussion of these criteria within the SEM framework is provided by \cite{raftery1995bayesian, haughton1997information}. Despite their popularity in structural equation modeling, only a limited number of studies have evaluated the performance of AIC, BIC and their extensions for model selection \citep{bollen2014bic, lin_selecting_2017}. However, these studies focus on situations in which a latent relationship is truly present and the objective is to assess the ability of the criteria to detect this signal. Much less attention has been paid to the complementary situation where no relationship exists between latent variables and the true model corresponds to a simpler structure. From a model selection perspective, correctly identifying the absence of a signal is equally important, as overfitting may lead to the detection of spurious latent dependencies. Consequently, a comprehensive evaluation of model selection criteria should consider both their ability to recover existing relationships and their ability to avoid selecting unnecessary ones when no signal is present.

The popularity of BIC stems from its interpretation as an asymptotic approximation of the integrated observed-data likelihood, making it a natural tool for model comparison. Nevertheless, SEMs present an additional level of complexity due to the presence of latent variables. Since these variables are not observed, the marginal likelihood underlying BIC is obtained after integrating out both the model parameters and the latent structure. As a result, the latent variable structure itself does not explicitly contribute to the model selection criterion. This motivates the investigation of alternative approaches based on the integrated complete-data likelihood, where the latent structure is directly taken into account. This naturally raises the following question: can model selection in SEM benefit from considering the integrated complete-data likelihood instead of the integrated observed-data likelihood alone?


In this work, we investigate the use of the Integrated Completed Likelihood (ICL, \cite{biernacki2010exact}) criterion for model selection in Gaussian SEM. We derive a closed-form expression of the integrated complete-data likelihood under specific priors and study its practical implementation. Because the integrated complete likelihood is then tractable, we also propose an approximation strategy of the integrated observed-data likelihood based on importance sampling, allowing us to construct a new criterion relying on the integrated complete-data likelihood. So, in this article, we revisit the model selection problem in SEM from this perspective. A comprehensive evaluation of these two new approaches is compared to classical model selection criteria in a series of simulation experiments involving both signal and no-signal scenarios. We investigate whether explicitly accounting for the latent variable structure through the integrated complete-data likelihood can improve model selection performance.

This paper is organized as follows. Section \ref{sec_model} describes the structural equation model to introduce the notation. Section \ref{sec_model_select} details our model selection strategy with the two proposed criteria. Next, Section \ref{sec_simu} evaluates the selection performance of our methods through an intensive simulation study (with signal and no-signal scenarios) and presents comparisons with existing methods. Finally, Section \ref{sec_conclusion} concludes with a summary discussion and prospects for future research. The proof of the main proposition is postponed to Appendix \ref{app_proof}.

\section{Model description}
\label{sec_model}

\subsection{Structural Equation Model (SEM)}
\label{sec_SEM}

Let $X_i \in \mathbb{R}^p$ denote the observed variables and $Z_i \in \mathbb{R}^q$ the latent variables associated with individual $i \in \{1,\dots,n\}$. We consider the following Gaussian Structural Equation Model (SEM). The measurement model, which defines the manifest variables $X_i$ conditionally on $Z_i$, satisfies: 
\[
X_i = \Lambda Z_i + \varepsilon_i, \quad \varepsilon_i \sim \mathcal{N}_p(0,\Sigma)
\]
where $\Lambda\in \mathbb{R}^{p\times q}$ is the loading matrix, and $\Sigma$ is a diagonal covariance matrix:
\[
\Sigma = \mathrm{diag}(\sigma_1^2,\dots,\sigma_p^2).
\] 
The structural model, linking the latent variables to one another, is defined as:
\[
Z_i = (I_q - B)^{-1} \xi_i, \quad \xi_i \sim \mathcal{N}_q(0, \Gamma)
\]
where $I_q$ is the identity matrix of size $q \times q$, $B \in \mathbb{R}^{q\times q}$ is the matrix that describes the dependencies between latent variables, and $\Gamma$ is a diagonal covariance matrix. To ensure that the model is identifiable, we assume that $B$ is strictly lower triangular (to ensure that the latent dependency graph is acyclic) and that $\Gamma = I_q$.

The sparsity patterns of $\Lambda$ and $B$ define the SEM structure. In particular, each observed variable is assumed to depend on a single latent variable only, meaning that each row of $\Lambda$ contains exactly one non-zero coefficient. We consider a sparse version of $B$ where the non-zero parameters are given by $\bm{f} \in \{0,1\}^{q \times q}$ such that 
\[
f_{h\ell} = 1 \iff B_{h\ell} \neq 0,
\quad f_{h\ell} = 0 \text{ if } h \le \ell.
\]
Let $d_h = \sum_{\ell=1}^q f_{h\ell}$ be the number of free parameters for row $h$ of $B$, and $\bm{m}_h \in \mathbb{R}^{d_h}$ the $d_h$-dimensional vector composed of the non-zero elements in row $h$ of $B$. Likewise, let $\bm{\omega}=(\bm{\omega}_1, \dots, \bm{\omega}_p)^\top \in \mathbb{R}^{p\times q}$ be the binary indicator matrix that defines sparsity in $\Lambda$:
\[
\bm{\omega}_{jh} = 1 \iff \Lambda_{jh} \neq 0,
\quad \bm{\omega}_{jh} = 0 \text{ otherwise}.
\]

The parameters to be estimated are $\theta = \{\Lambda, \Sigma, B\}$. 

\subsection{Observed-data distribution and parameter estimation}
\label{sec_estimation}

Under the previous assumptions, the latent variables follow 
\[Z_i \mid \theta \sim \mathcal{N}_q(0,\Psi)\] with
\[
\Psi = (I_q - B)^{-1} \left[(I_q - B)^{-1}\right]^\top.
\]
Consequently, the observed variables satisfy:
\begin{equation}
    \label{eq_likelihood}
    X_i \mid \theta \sim \mathcal{N}_p(0,S_{11}),
\end{equation}
where
\begin{equation}
\label{eq_S11}
    S_{11} = \Sigma + \Lambda \Psi \Lambda^\top.
\end{equation}

Parameter estimation is performed through maximum likelihood estimation. 
The classical LISREL approach estimates the parameter vector $\theta = (\Lambda, B, \Sigma)$ by minimizing the discrepancy between the empirical covariance matrix and the theoretical covariance matrix implied by the model \citep{joreskog1970general}.

In practice, the optimization is carried out numerically using standard SEM software such as the \texttt{lavaan} package in \texttt{R}.

\section{Model selection in SEM}
\label{sec_model_select}

In the following, we consider a $n$-sample $X = (X_1, \dots, X_n)^\top \in \mathbb{R}^{n \times p}$ of observed variables and $Z~=~(Z_1,\dots,Z_n)^\top\in \mathbb{R}^{n \times q}$ the latent variables. We denote the elements of the matrices $X$ and $Z$ by $(x_{ij})_{\substack{1\leq i \leq n \\ 1 \leq j \leq p}}$ and $(z_{ih})_{\substack{1\leq i \leq n \\ 1 \leq h \leq q}}$, respectively.

\subsection{Bayesian Information Criterion : approximate integrated observed-data likelihood}

Model selection in SEM aims to determine the graph structure encoded by the sparsity patterns of $B$. Given a collection of candidate models $(\mathcal{M}_g)_g$, the objective is to select the model providing the best compromise between goodness-of-fit and model complexity.

Most approaches rely on penalized likelihood criteria of the form:
\begin{equation}
\label{eq_IC}
    \mathrm{Crit}(\mathcal{M}) = \log p_\mathcal{M}(X \mid \hat \theta_\mathcal{M}) - \frac{1}{2} \mathrm{pen}(\mathcal{M})
\end{equation}
where $p_\mathcal{M}(X \mid \hat{\theta}_\mathcal{M})$ denotes the likelihood evaluated at the maximum likelihood estimator under model $\mathcal{M}$, and $\mathrm{pen}(\mathcal{M})$ is a penalty term of the model $\mathcal{M}$. The goal is to select the model that maximizes this criterion. The most widely used criterion is the Bayesian Information Criterion (BIC) \citep{schwarz1978estimating}:
\[BIC(\mathcal{M}) = \log p_\mathcal{M}(X \mid \hat \theta_\mathcal{M}) - \frac{k_\mathcal{M}}{2} \log(n),\]
where $k_\mathcal{M}$ is the number of free parameters in model $\mathcal{M}$ and $n$ is the sample size. It should be noted that the BIC was obtained, under the assumption of a uniform prior distribution over the candidate models, from a Bayesian approximation of the integrated observed-data likelihood $p(X \mid \mathcal{M})$ 
\[p(X \mid \mathcal{M}) = \int_\Theta p_\mathcal{M}(X \mid \theta) p(\theta) d\theta,\]
using Laplace's method:
\begin{equation}
\label{eq_laplace}
  \log p(X \mid M) \approx 
\log p_\mathcal{M}(X \mid \hat{\theta}_\mathcal{M}) - \frac{k_\mathcal{M}}{2} \log n.
\end{equation}

\subsection{Exact integrated complete-data likelihood}

\subsubsection{Motivation}

In SEM, since the likelihood \eqref{eq_likelihood} has a closed-form expression, it is possible to explicitly calculate standard information criteria like \eqref{eq_IC}. These standard criteria, such as BIC, are based solely on the integrated observed-data likelihood $p(X \mid \mathcal{M})$. However, in SEM, the latent variables contain structural information regarding the network of dependencies that we wish to identify. This suggests that model selection may benefit from considering the integrated complete-data likelihood $p(X, Z \mid \mathcal{M})$ instead of the integrated observed-data likelihood alone $p(X \mid \mathcal{M})$.

\subsubsection{Integrated Completed Likelihood (ICL) criterion}

In a Bayesian framework, the integrated complete-data likelihood for latent variable models is defined as:
\begin{equation}
\label{eq_Int_ICL}
    p(X,Z \mid \mathcal{M})
=
\int p(X,Z\mid \theta ;\mathcal{M})p(\theta\mid \mathcal{M})d\theta.
\end{equation}
where $p(\theta\mid \mathcal{M})$ denotes the prior distribution on $\theta$.

For many latent variable models, the exact computation of this integral is analytically intractable. A classical approximation is obtained using Laplace's method, as for BIC criterion, leading to the asymptotic expression
\[
\log p(X,Z \mid \mathcal{M})
=
\log p(X,Z;\widehat{\theta}_\mathcal{M})
-
\frac{k_\mathcal{M}}{2}\log n
+
O_p(1),
\]
where $k_\mathcal{M}$ denotes the number of free parameters in model $\mathcal{M}$, $n$ is the sample size, and $\widehat{\theta}_\mathcal{M}$ is the maximum likelihood estimator obtained from the observed data in model $\mathcal{M}$. Since the latent variables are not observed, they are replaced by their Maximum A Posteriori (MAP) estimates $\widehat{Z}$. The resulting criterion is the Integrated Completed Likelihood (ICL), called ICLbic criterion \citep{biernacki2000assessing, biernacki2010exact}:
\[
\mathrm{ICLbic}
=
\log p(X,\widehat{Z};\widehat{\theta}_\mathcal{M})
-
\frac{k_\mathcal{M}}{2}\log n.
\]

Unlike BIC, which only relies on the observed-data likelihood, ICLbic incorporates the latent structure through the complete-data likelihood. As a consequence, in the classification context in which ICLbic was introduced, ICLbic generally favors models associated with clearly distinct latent structures.

Fortunately, in the Gaussian SEM framework considered in this work, the integrated complete-data likelihood \eqref{eq_Int_ICL} can be derived explicitly under suitable conjugate prior distributions. As shown in the next sections, this allows us to obtain a closed-form expression for $p(X,Z \mid \mathcal{M})$ and therefore to compute the corresponding ICL criterion exactly rather than through an asymptotic approximation.

\subsubsection{Prior specification}
\label{sec_priors}

Considering a Bayesian framework for the SEM model presented in section \ref{sec_SEM}, the prior distribution of $\theta$ is defined by

\[
\pi(\theta \mid \vartheta, \bm{f}, \bm{\omega})
=
\pi(\Lambda,\Sigma \mid \vartheta,\bm{\omega})
\pi(B \mid \vartheta,\bm{f}),
\]

with $\vartheta$ denotes all the hyperparameters, and where $\bm{f}$ and $\bm{\omega}$ encode the sparsity structures of $B$ and $\Lambda$, respectively.

Recall that only one element of each line of $\Lambda$ is not zero. We consider that the prior distribution of non-zero elements of $\Lambda$ is a product of independent Gaussian distributions given $\Sigma$ and that the prior distribution of the diagonal elements of $\Sigma$ is a product of independent inverse gamma distributions, so that

\begin{align}
\pi(\Lambda,\Sigma \mid \vartheta,\bm{\omega})
=
\prod_{j=1}^p
\Bigg(
g(\sigma_j^2; \alpha_j/2, \beta_j^2/2)
\prod_{h=1}^q
\left[
\phi_1(\Lambda_{jh}; \nu_j, \sigma_j^2 \delta_j^{-1})
\right]^{\omega_{jh}}
\Bigg),
\end{align}
where $g(x; a,b)=b^a(1/x)^{a+1}\exp(-b/x)/\Gamma(a)$ is the density of an inverse gamma distribution, for some hyperparameters $\alpha_j>0$, $\beta^2_j>0$, $\delta_j>0$.

For $B$, we also consider independent gaussian distributions on the non-zero elements: we have the following prior

\[\pi(B \mid \vartheta,\bm{f}) =\prod_{h=1}^q
\phi_{d_h}(\bm{m}_h; \bm{\mu}_h, \bm{\kappa}_h^{-1}), \]
where $\bm{\mu}_h \in \mathbb{R}^{d_h}$ and $\bm{\kappa}_h \in \mathbb{R}^{d_h \times d_h}$.

\subsubsection{Exact ICL in SEM : a closed-form integrated complete-data likelihood}

\begin{prop}
\label{propICL}
In the SEM model presented in Section \ref{sec_SEM} and with the priors presented in Section \ref{sec_priors}, the integrated complete-data likelihood is equal to
\begin{align*}
    p(X,Z \mid \vartheta,\bm{f}, \bm{\omega})
=
&\frac{1}{\pi^{(p+q)n/2} 2^{qn/2}}
\prod_{h=1}^q \exp \left\{ -\frac{1}{2} \sum_{i=1}^n z_{ih}^2 \right\}
\left[
\frac{\det^{1/2}(\bm{\kappa}_h)}{\det^{1/2}(\tilde S_h)}
\exp\left(-\frac{\tilde t_h}{2}\right)
\right]^{\mathds{1}_{\{d_h>0\}}} \times \\
&\prod_{j=1}^p
\left[
\frac{\delta_j^{1/2}}{s_j^{1/2}}
\frac{\Gamma(n/2+\alpha_j/2)}{\Gamma(\alpha_j/2)}
\frac{\beta_j^{\alpha_j}}
{(\beta_j^2+t_j)^{n/2+\alpha_j/2}}
\right].
\end{align*}
where $t_j=\sum_{i=1}^n x_{ij}^2 + \delta_j \nu_j^2 - s_j a_j^2$, $s_j=\delta_j + \bm{\omega}_j^\top \sum_{i=1}^n Z_i Z_i^\top \bm{\omega}_j$, $a_j=s_j^{-1} (\sum_{i=1}^n \bm{\omega}_j^\top Z_i x_{ij} + \delta_j \nu_j)$, $\tilde t_h = \bm{\mu}_h^\top \bm{\kappa}_h \bm{\mu}_h - \bm{\zeta}_h^\top \tilde S_h \bm{\zeta}_h$, $\tilde S_h = \tilde M_h + \bm{\kappa}_h$, $\bm{\zeta}_h = \tilde S_h^{-1}(\bm{\tilde m}_h +\bm{\kappa}_h \bm{\mu}_h)$, $\tilde M_h$ is the symmetric matrix of dimension $d_h \times d_h$ composed of rows and columns of matrix $\sum_{i=1}^n Z_i Z_i^\top$ having a index $\ell$ such that $f_{h\ell}=1$ and $\bm{\tilde m}_h$ is the $d_h$-dimensional vector composed of the elements $\ell$ of
row $h$ of $\sum_{i=1}^n z_{ih} Z_i$ such that $f_{h\ell}=1$. 
\end{prop}

The proof is available in Appendix \ref{app_proof}.

Thus, as for the calculation of ICLbic, we replace $Z$ with its MAP estimator $\hat Z$ to define the ICL criterion as:
\begin{equation}
\label{eq_ICL}
\mathrm{ICL}(\mathcal{M})
=
\log p(X,\hat{Z}\mid \mathcal{M}),
\end{equation}
where $\mathcal{M}$ is defined by $\bm{f}$ and $\bm{\omega}$. The selected model is then obtained by maximizing the resulting criterion. The following paragraph details the calculation of the MAP estimator of $Z$.

\subsubsection*{Latent estimation}

In this SEM model it is easy to sample $Z$ from its conditional distribution $Z \mid X=x, \theta$ for any parameter $\theta$. Indeed, by the properties of Gaussian vectors, the joint distribution of $(X, Z)$ given $\theta$ is a centered $(p+q)$-dimensional Gaussian distribution with covariance
\[
S =
\begin{pmatrix}
S_{11} & S_{12} \\
S_{21} & \Psi
\end{pmatrix} \in \mathbb{R}^{(p+q) \times (p+q)}
\]
where $S_{11}$ is given by equation \eqref{eq_S11}, 
\[
S_{12} = \Lambda \Psi,
\quad
S_{21} = S_{12}^\top.
\]
So, the conditional distribution of $Z \mid X=x, \theta$ is Gaussian with mean $\upsilon_{\theta}(x)$ and covariance matrix $\Upsilon_{\theta}$ where

\[
\upsilon_{\theta}(x) = S_{21} S_{11}^{-1} x,
\quad
\Upsilon_{\theta} = \Psi - S_{21} S_{11}^{-1} S_{12}.
\]

Thus, the MAP estimator $\hat Z$ of $Z$ that maximizes its posterior probability is equal to its posterior expectation, which yields:
\[
\hat Z (x)= \arg\max_Z p(Z|X=x, \theta) = \upsilon_{\theta}(x).
\]

\subsection{Importance Sampling approximation of the integrated observed-data likelihood}

Now that the integrated complete-data likelihood can be computed analytically thanks to Proposition \ref{propICL}, it becomes possible to exploit this expression in order to approximate the integrated observed-data likelihood. Unlike the standard BIC approximation, which relies on a Laplace expansion around the maximum likelihood estimator, the proposed approach directly incorporates the latent structure of the SEM into the estimation procedure.

More precisely, we aim at approximating the integrated observed-data likelihood
\[
p(X \mid \mathcal{M})=\int p(X,Z \mid \mathcal{M})\,dZ,
\]
by taking advantage of the explicit expression previously derived for the integrated complete-data likelihood $p(X,Z \mid \mathcal{M})$. To achieve this, we rely on an importance sampling strategy. By introducing a proposal distribution $h(Z)$, we obtain
\[
p(X \mid \mathcal{M})
=
\int \frac{p(X,Z \mid \mathcal{M})}{h(Z)}h(Z)dZ.
\]

An unbiased Monte Carlo approximation is therefore given by
\[
p(X \mid \mathcal{M})
\approx
\frac{1}{R}\sum_{r=1}^R
\frac{p(X,Z^{(r)}\mid \mathcal{M})}{h(Z^{(r)})},
\]
where $Z^{(1)},\dots, Z^{(R)}$ is the sample set from the distribution $h$. To reduce numerical instability, computations are performed on the logarithmic scale using a centering strategy based on the dominant importance weight:
\[Z^* = \underset{{Z^{(r)},  1 \leq r \leq R} }{\text{argmax}} \dfrac{p(X,Z^{(r)}\mid \mathcal{M})}{h(Z^{(r)})}.\]
More precisely, we approximate 
\[\log p(X \mid \mathcal{M}) \approx \log\left( \dfrac{p(X,Z^*\mid \mathcal{M})}{h(Z^*)} \right) + \log \left[ \dfrac{1}{R} \sum_{r=1}^R \exp \left(\log\left( \dfrac{p(X,Z^{(r)}\mid \mathcal{M})}{h(Z^{(r)})} \right) - \log\left( \dfrac{p(X,Z^*\mid \mathcal{M})}{h(Z^*)} \right) \right) \right]. \]

A natural choice for the proposal distribution is the posterior distribution of the latent variables under the fitted SEM:
\[
h(Z)=p(Z\mid X,\widehat{\theta}),
\]
which is Gaussian $\mathcal{N}_q(\upsilon_{\hat\theta}(X),\Upsilon_{\hat\theta})$ in our setting. 

Therefore, in what follows, we will use the term ``logIL.IS criterion`` (log Integrated Likelihood - Importance Sampling) to refer to the model selection strategy that consists of maximizing the following expression:
\begin{equation}
\label{eq_BIC_IS}
    \mathrm{logIL.IS}(\mathcal{M}) = \log\left( \dfrac{p(X,Z^*\mid \mathcal{M})}{h(Z^*)} \right) + \log \left[ \dfrac{1}{R} \sum_{r=1}^R  \exp \left(\log\left( \dfrac{p(X,Z^{(r)}\mid \mathcal{M})}{h(Z^{(r)})} \right) - \log\left( \dfrac{p(X,Z^*\mid \mathcal{M})}{h(Z^*)} \right) \right) \right].
\end{equation}

\section{Numerical experiments}
\label{sec_simu}

This section examines the model selection performance of the two methods proposed in this paper, ICL \eqref{eq_ICL} and logIL.IS \eqref{eq_BIC_IS}, in the Gaussian SEM setting defined in Section \ref{sec_SEM}, and compares them with existing state-of-the-art approaches.

\subsection{Simulation design}
\label{sec_simu_design}

The data are generated according to the following SEM model: for $1\leq i \leq n$ , 

\begin{equation*}
    \left\{
    \begin{array}{ll}
        X_i = \Lambda Z_i + \varepsilon_i, \quad &\varepsilon_i \sim \mathcal{N}_{p}(0,\Sigma), \quad \Sigma = \mathrm{diag}(\sigma_1^2, \dots, \sigma_p^2), \\
        Z_i = B Z_i + \xi_i, \quad &\xi_i \sim \mathcal{N}_{q}(0, I_q),
    \end{array}
\right.
\end{equation*}
with $q=3$ latent variables, $p=12$ observed variables, corresponding to four manifest variables per latent variable, and $B$ and $\Lambda$ of the following form
\[B=\begin{pmatrix}
0      & 0      & 0 \\
b_{21}  & 0      & 0 \\
b_{31}  & b_{32}  & 0
\end{pmatrix}, \quad \Lambda =
\left(
\begin{array}{cccccccccccc}
1 & 1 & 1 & 1 & 0 & 0 & 0 & 0 & 0 & 0 & 0 & 0 \\
0 & 0 & 0 & 0 & 1 & 1 & 1 & 1 & 0 & 0 & 0 & 0 \\
0 & 0 & 0 & 0 & 0 & 0 & 0 & 0 & 1 & 1 & 1 & 1
\end{array}
\right)^\top,\]
where different latent dependency structures are investigated through the coefficients of the structural matrix $B$ and the loading matrix is chosen such that each observed variable depends on a single latent variable. 

To test the sensitivity of the methods to different parameter values, we test various simulation scenarios: different sample sizes $n\in \{70,250,1000\}$, several signal strengths for the latent effects $b_{31}\in\{0,0.1,0.4\}$, $b_{21}=b_{32} \in \{0,\sqrt{0.1},\sqrt{0.4}\}$, and different levels of residual error $\sigma_1^2=\dots=\sigma_p^2\in\{0.1,0.3\}$. For each configuration, $100$ datasets are simulated and the following candidate models are fitted and compared using ICL, logIL.IS and several model selection criteria presented in the next section. These models are represented in Figure \ref{fig:mod_collection}. The first model, called "complete model", corresponds to the case where $b_{31} \neq 0$ and $b_{21}=b_{32}\neq 0$. The second model, called "indirect model", corresponds to the case where $b_{31} = 0$ but $b_{21}=b_{32}\neq 0$. The "direct model" corresponds to the case where $b_{31} \neq 0$ but $b_{21}=b_{32}= 0$. And the null model corresponds to the case where $b_{31} = b_{21}=b_{32} =0$. This collection of models is inspired by a problem in the field of agroecology where the goal is to determine whether the variable $z_1$ is directly related to $z_3$, or only through $z_2$, or both, or has no relationship to it at all.

\begin{figure}[h!]
    \centering
    
    \begin{subfigure}{0.45\textwidth}
        \centering
        \includegraphics[width=\linewidth]{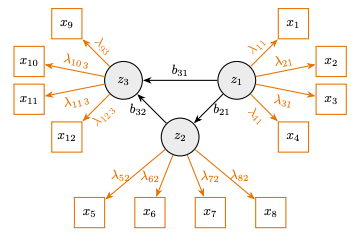}
        \caption{Complete model}
        \label{fig:Comp_Mod}
    \end{subfigure}
    \hfill
    \begin{subfigure}{0.45\textwidth}
        \centering
        \includegraphics[width=\linewidth]{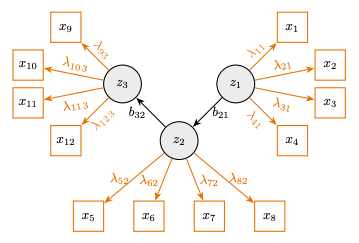}
        \caption{Indirect model}
        \label{fig:Ind_Mod}
    \end{subfigure} \\

    \begin{subfigure}{0.45\textwidth}
        \centering
        \includegraphics[width=\linewidth]{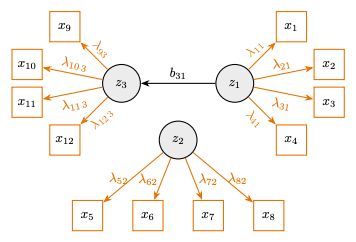}
        \caption{Direct model}
        \label{fig:Dir_Mod}
    \end{subfigure}
    \hfill
    \begin{subfigure}{0.45\textwidth}
        \centering
        \includegraphics[width=\linewidth]{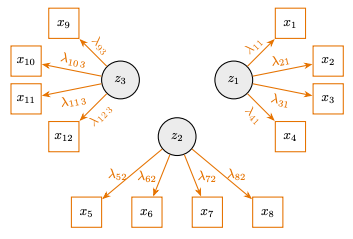}
        \caption{Null model}
        \label{fig:Null_Mod}
    \end{subfigure} \\
    \caption{Models collection}
    \label{fig:mod_collection}
\end{figure}

\subsection{Model selection criteria under comparison}

To assess the performance of the two proposed approaches, namely ICL and logIL.IS, we compare them with several classical and refined likelihood-based model selection criteria commonly used in the literature. This section briefly recalls the considered criteria and specifies the hyperparameter choices used for the implementation of the proposed methods.

The two most commonly used model selection criteria are the Akaike Information Criterion (AIC, \cite{akaike1974new}) and the Bayesian Information Criterion (BIC, \cite{schwarz1978estimating}). Both criteria are derived within the maximum likelihood estimation framework described in Equation~\eqref{eq_IC}. AIC aims at minimizing the expected Kullback-Leibler divergence between the fitted and the true model, whereas BIC can be interpreted as an asymptotic approximation of the integrated observed-data likelihood. 
An important distinction between the two criteria is that BIC is asymptotically consistent for model selection under suitable regularity conditions, while AIC is not. Several extensions and refinements of these criteria have subsequently been proposed in order to improve finite-sample behavior or provide more accurate approximations of the marginal likelihood. For instance, the Consistent Akaike Information Criterion (CAIC, \cite{bozdogan1987model}) modifies the AIC penalty in order to recover asymptotic consistency. Note that CAIC is indeed an extension of AIC, even though its penalty is very similar to that of BIC. The Adjusted Bayesian Information Criterion (ABIC, \cite{sclove1987application}) was introduced from a minimum description length perspective and has shown strong empirical performance in selecting the correct number of factors and latent classes \citep{dziak2020sensitivity,yang2006evaluating}. Other variants, such as HBIC (Haughton BIC, \cite{haughton1988choice}) and IBIC (Information matrix-based BIC, \cite{bollen2012comparison}), can be viewed as refined asymptotic versions of BIC. 
These criteria are obtained by retaining additional terms in the Laplace approximation of the integrated likelihood beyond the standard BIC expansion in Equation~\eqref{eq_laplace}. 
Originally introduced in the context of exponential family models, HBIC provides a theoretically refined approximation of the Bayes factor. Similarly, IBIC incorporates information from the observed Fisher information matrix and was shown by \cite{bollen2012comparison} to improve the accuracy of model recovery over standard BIC in small samples.

The expressions of the considered criteria are recalled below:
\begin{itemize}
    \item $AIC(\mathcal{M}) = \log p_\mathcal{M}(X \mid \hat \theta_\mathcal{M}) - k_\mathcal{M}$, \citep{akaike1974new},
    
    \item $CAIC(\mathcal{M}) = \log p_\mathcal{M}(X \mid \hat \theta_\mathcal{M}) - \dfrac{k_\mathcal{M}}{2}(\log(n)+1)$, \citep{bozdogan1987model},
    
    \item $BIC(\mathcal{M}) = \log p_\mathcal{M}(X \mid \hat \theta_\mathcal{M}) - \dfrac{k_\mathcal{M}}{2}\log(n)$, \citep{schwarz1978estimating},
    
    \item $ABIC(\mathcal{M}) = \log p_\mathcal{M}(X \mid \hat \theta_\mathcal{M}) - \dfrac{k_\mathcal{M}}{2} \log\left(\dfrac{n+2}{24}\right)$, \citep{sclove1987application},
    
    \item $HBIC(\mathcal{M}) = \log p_\mathcal{M}(X \mid \hat \theta_\mathcal{M}) - \dfrac{k_\mathcal{M}}{2} \log\left(\dfrac{n}{2\pi}\right)$, \citep{haughton1988choice},
    
    \item $IBIC(\mathcal{M}) = \log p_\mathcal{M}(X \mid \hat \theta_\mathcal{M}) - \dfrac{k_\mathcal{M}}{2} \log\left(\dfrac{n}{2\pi}\right) -\dfrac{1}{2} \log\left(\det I(\hat{\theta}_\mathcal{M})\right)$, \citep{bollen2012comparison},
\end{itemize}
where $k_\mathcal{M}$ denotes the number of free parameters of model $\mathcal{M}$ and $I(\hat{\theta}_\mathcal{M})$ is the Fisher information matrix evaluated at the maximum likelihood estimator. Note that here $k_{\mathcal{M}}$ equals $27$, $26$, $25$, and $24$ for the complete, indirect, direct, and null models, respectively.

It is important to emphasize that all these criteria are based exclusively on the observed-data likelihood. Consequently, they do not explicitly exploit the latent variable structure underlying the SEM. In contrast, the two approaches proposed in this work, namely ICL and logIL.IS, directly incorporate information from the complete-data representation involving the latent variables.

To compute the proposed ICL criterion, the following prior hyperparameter values are used:
\begin{itemize}
    \item $\alpha_j = 1$ for all $1 \leq j \leq p$,
    \item $\beta_j^2 = 1$ for all $1 \leq j \leq p$,
    \item $\nu_j = 0$ for all $1 \leq j \leq p$,
    \item $\delta_j = 2$ for all $1 \leq j \leq p$,
    \item $\bm{\mu}_h = (0,\dots,0)^\top$ for all $1 \leq h \leq q$,
    \item $\bm{\kappa}_h = \mathrm{diag}(2,\dots,2)$ for all $1 \leq h \leq q$.
\end{itemize}

For the importance sampling approximation involved in the logIL.IS criterion, the number of Monte Carlo samples is set to $R=n$.

\subsection{Results}

\subsubsection{All comparison results for $\sigma_j^2=0.1$}

The following results focus on the case where $\sigma_j^2=0.1$. In Figures \ref{fig:res_Indirect}--\ref{fig:res_null}, the different color shades are used solely to distinguish between model selection methods based on AIC (in purple), those based on BIC (in blue), and the two methods we propose (in orange).

\paragraph{Indirect model.}

\begin{figure}[h!]
    \centering
    
    \begin{subfigure}{0.45\textwidth}
        \centering
        \includegraphics[width=\linewidth]{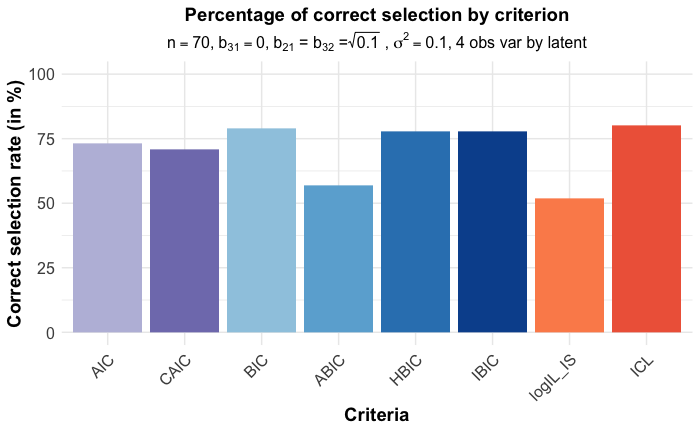}
        \caption{For $n=70$ and $b_{21}=b_{32}=\sqrt{0.1}$.}
    \end{subfigure}
    \hfill
    \begin{subfigure}{0.45\textwidth}
        \centering
        \includegraphics[width=\linewidth]{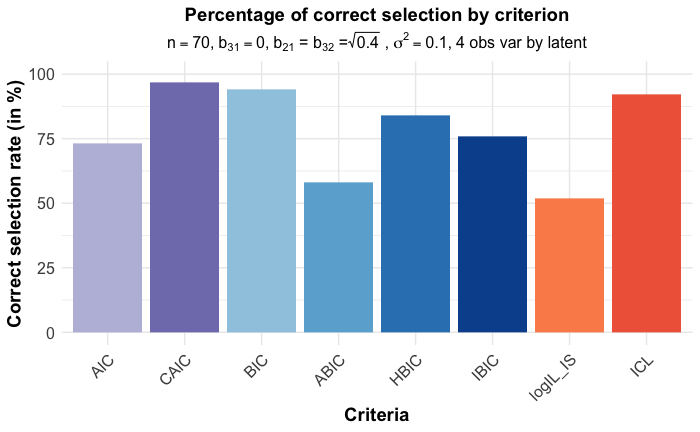}
        \caption{For $n=70$ and $b_{21}=b_{32}=\sqrt{0.4}$.}
    \end{subfigure}
    
    \vspace{0.5cm}
    
    \begin{subfigure}{0.45\textwidth}
        \centering
        \includegraphics[width=\linewidth]{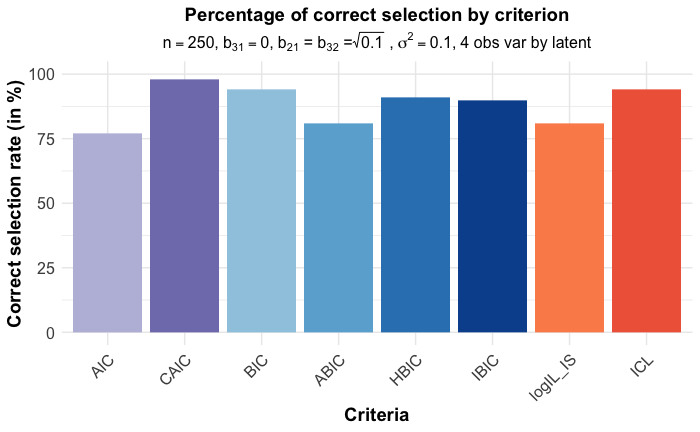}
        \caption{For $n=250$ and $b_{21}=b_{32}=\sqrt{0.1}$.}
        \label{fig:NoMed_c}
    \end{subfigure}
    \hfill
    \begin{subfigure}{0.45\textwidth}
        \centering
        \includegraphics[width=\linewidth]{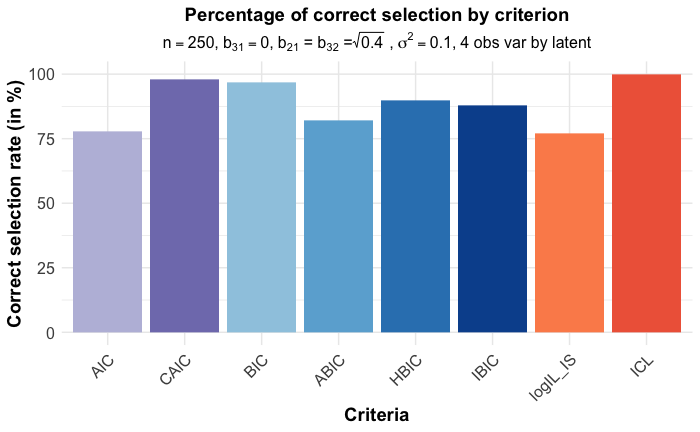}
        \caption{For $n=250$ and $b_{21}=b_{32}=\sqrt{0.4}$.}
    \end{subfigure}

    \vspace{0.5cm}
    
    \begin{subfigure}{0.45\textwidth}
        \centering
        \includegraphics[width=\linewidth]{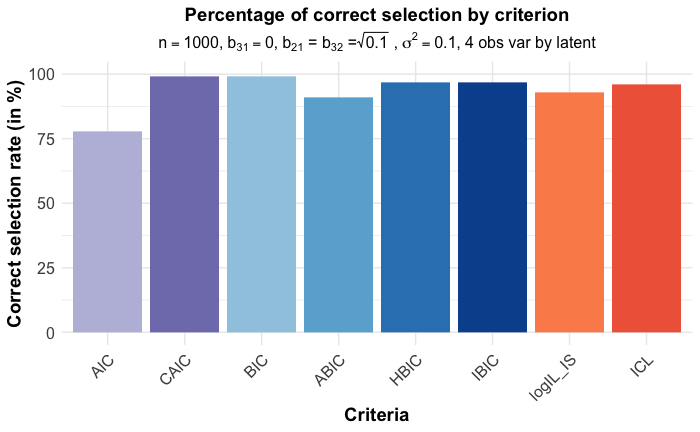}
        \caption{For $n=1000$ and $b_{21}=b_{32}=\sqrt{0.1}$.}
    \end{subfigure}
    \hfill
    \begin{subfigure}{0.45\textwidth}
        \centering
        \includegraphics[width=\linewidth]{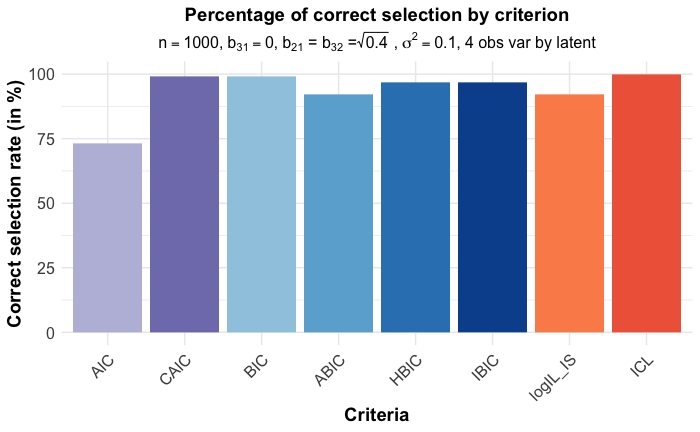}
        \caption{For $n=1000$ and $b_{21}=b_{32}=\sqrt{0.4}$.}
    \end{subfigure}
    
    \caption{Results of correct selection when the true model is the indirect model.}
    \label{fig:res_Indirect}
\end{figure}

Figure~\ref{fig:res_Indirect} reports the percentage of correct model selection when the true data-generating process corresponds to the indirect model. Overall, model recovery improves substantially with the sample size. For $n=70$, the criteria show heterogeneous performance. For $n=1000$, almost all criteria achieve high correct selection rates.

When the indirect effects are strong, CAIC, BIC and ICL provide the highest selection rates, whereas ABIC and logIL.IS perform the worst especially in small samples. When indirect effects are weaker, CAIC and BIC are strongly influenced by the weakness of the signal for small sample sizes, whereas the ICL is less affected and becomes the best-performing criterion under these conditions for small $n$. Indeed, while ICL does not consistently outperform the other criteria, it remains very competitive in all configurations when we want to detect the indirect model.

logIL.IS provides stable performance across both configurations. Although it is not the best-performing criterion in small samples, it remains competitive as the sample size increases, which is not the case of AIC. 

\paragraph{Complete model.}

\begin{figure}[h!]
    \centering
    
    \begin{subfigure}{0.45\textwidth}
        \centering
        \includegraphics[width=\linewidth]{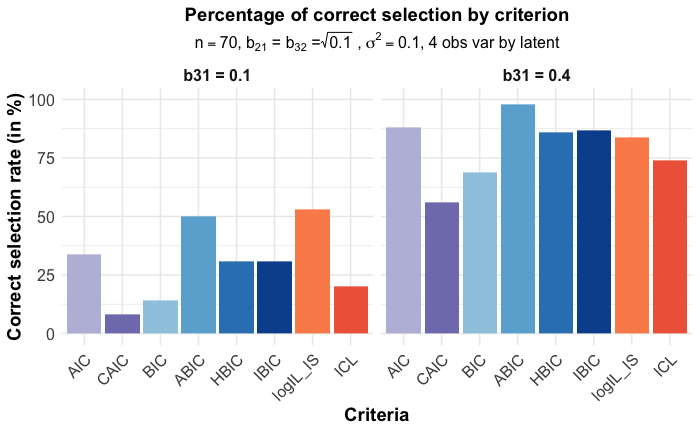}
        \caption{For $n=70$ and $b_{21}=b_{32}=\sqrt{0.1}$.}
    \end{subfigure}
    \hfill
    \begin{subfigure}{0.45\textwidth}
        \centering
        \includegraphics[width=\linewidth]{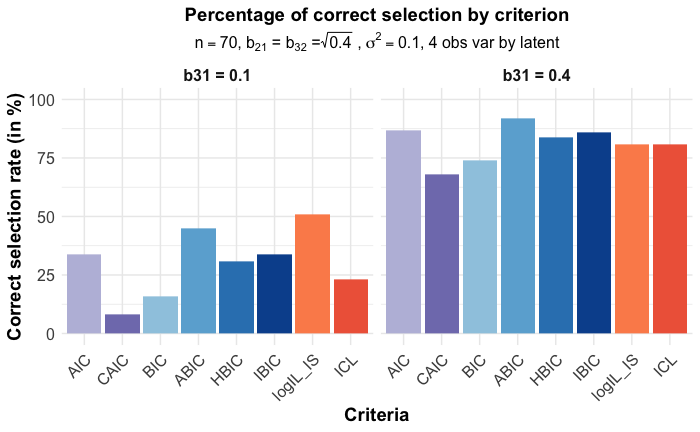}
        \caption{For $n=70$ and $b_{21}=b_{32}=\sqrt{0.4}$.}
    \end{subfigure}
    
    \vspace{0.5cm}
    
    \begin{subfigure}{0.45\textwidth}
        \centering
        \includegraphics[width=\linewidth]{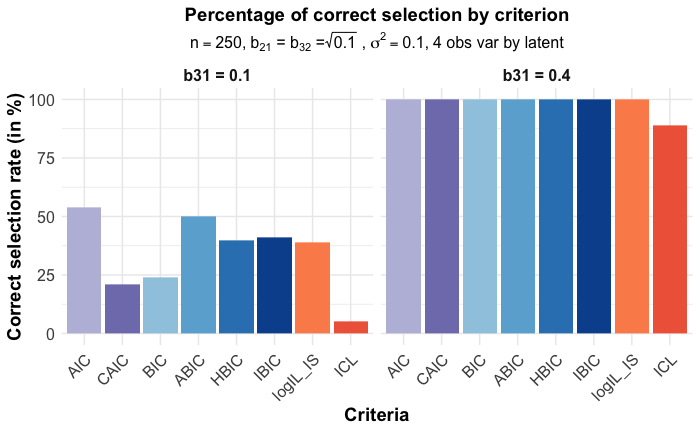}
        \caption{For $n=250$ and $b_{21}=b_{32}=\sqrt{0.1}$.}
    \end{subfigure}
    \hfill
    \begin{subfigure}{0.45\textwidth}
        \centering
        \includegraphics[width=\linewidth]{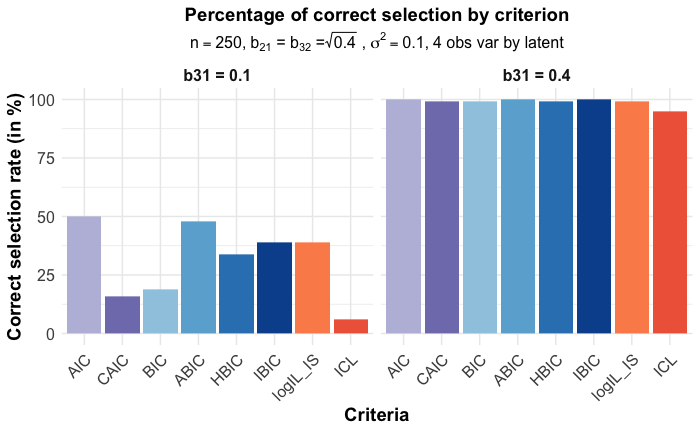}
        \caption{For $n=250$ and $b_{21}=b_{32}=\sqrt{0.4}$.}
    \end{subfigure}

    \vspace{0.5cm}

    \begin{subfigure}{0.45\textwidth}
        \centering
        \includegraphics[width=\linewidth]{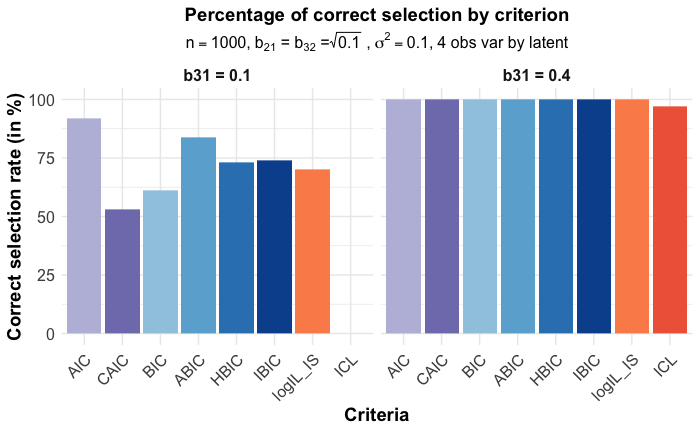}
        \caption{For $n=1000$ and $b_{21}=b_{32}=\sqrt{0.1}$.}
    \end{subfigure}
    \hfill
    \begin{subfigure}{0.45\textwidth}
        \centering
        \includegraphics[width=\linewidth]{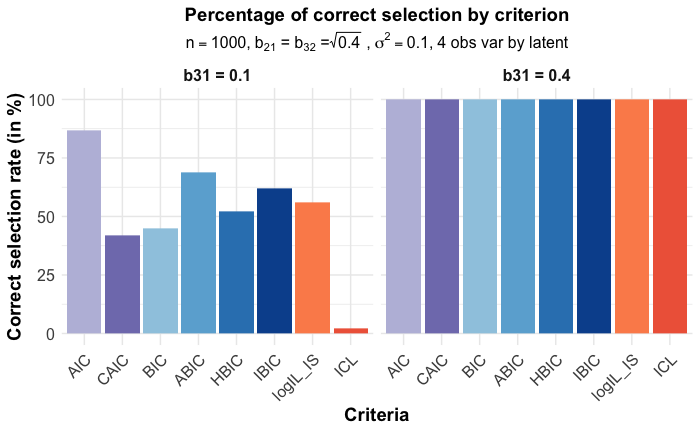}
        \caption{For $n=1000$ and $b_{21}=b_{32}=\sqrt{0.4}$.}
    \end{subfigure}
    
    \caption{Results of correct selection when the true model is the complete model.}
    \label{fig:res_Complete}
\end{figure}

Figure~\ref{fig:res_Complete} shows that recovery of the complete model is mainly driven by the magnitude of the direct effect $b_{31}$. When $b_{31}=0.1$, correct selection rates remain low in small samples, indicating that the complete structure is difficult to identify. In contrast, when $b_{31}=0.4$, all criteria improve substantially, and nearly perfect recovery is obtained for $n=250$ and $n=1000$.

In this setting, logIL.IS, ABIC, AIC, HBIC and IBIC tend to perform better when the direct effect is weak, whereas CAIC, BIC and ICL often underselect the complete model. logIL.IS is the best criterion for identifying the complete model when $n$ is small and $b_{31}$ is small. Moreover, logIL.IS performs better than BIC, suggesting that incorporating the indirect information structure helps reduce underselection. ICL, on the other hand, seems to have some difficulty identifying the complete pattern when $b_{31}$ is small (it selects the indirect model), but becomes competitive when the direct effect is large enough.

\paragraph{Direct model.}

\begin{figure}[h!]
    \centering
    
    \begin{subfigure}{0.45\textwidth}
        \centering
        \includegraphics[width=\linewidth]{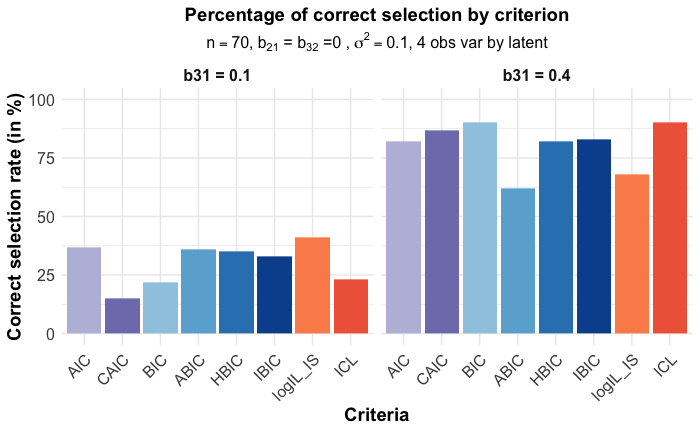}
        \caption{For $n=70$ and $b_{21}=b_{32}=0$.}
    \end{subfigure}
    \hfill
    \begin{subfigure}{0.45\textwidth}
        \centering
        \includegraphics[width=\linewidth]{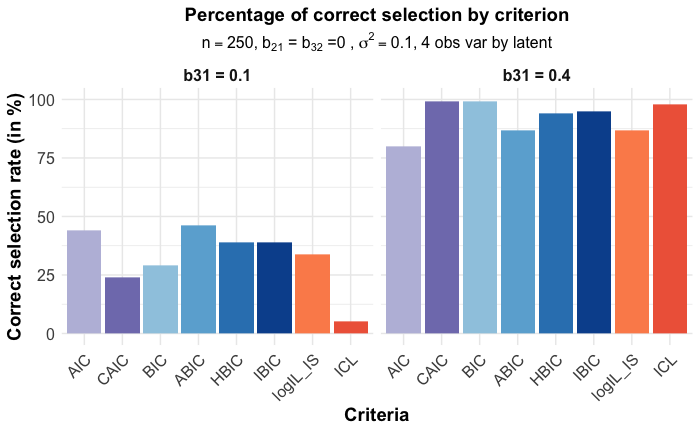}
        \caption{For $n=250$ and $b_{21}=b_{32}=0$.}
    \end{subfigure}

    \begin{subfigure}{0.45\textwidth}
        \centering
        \includegraphics[width=\linewidth]{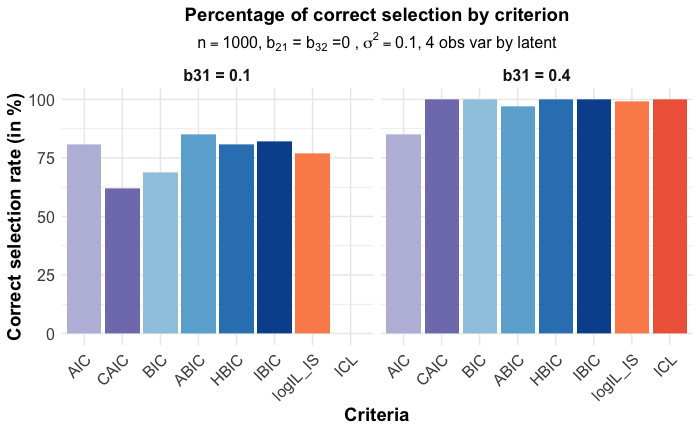}
        \caption{For $n=1000$ and $b_{21}=b_{32}=0$.}
    \end{subfigure}
    
    \caption{Results of correct selection when the true model is the Direct model.}
    \label{fig:res_Direct}
\end{figure}

Figure~\ref{fig:res_Direct} presents the results obtained when the true model is the direct model. As expected, both larger sample sizes and stronger direct effects improve model recovery. When the direct effect is weak, all criteria have relatively low correct selection rates, especially for $n=70$. When the direct effect is strong, the differences between criteria become much smaller, and almost all methods recover the true model accurately for $n$ large enough.

logIL.IS, AIC, ABIC, HBIC and IBIC are generally more effective in weak-signal situations, whereas BIC, ICL and CAIC become highly competitive when the direct effect increases. 

\paragraph{Null model.}

\begin{figure}[h!]
    \centering
    
    \begin{subfigure}{0.45\textwidth}
        \centering
        \includegraphics[width=\linewidth]{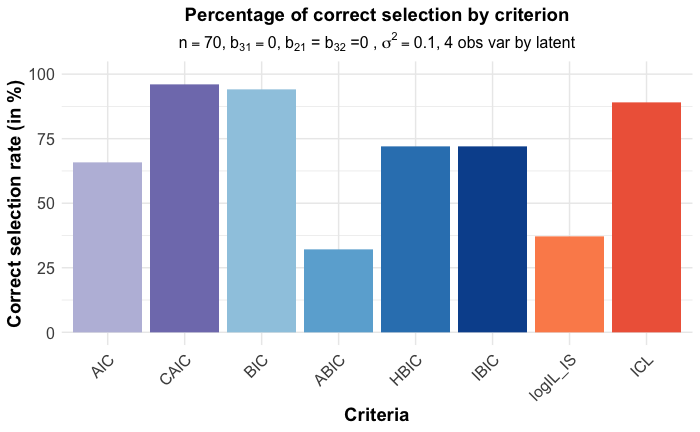}
        \caption{For $n=70$ and $b_{21}=b_{32}=0$.}
    \end{subfigure}
    \hfill
    \begin{subfigure}{0.45\textwidth}
        \centering
        \includegraphics[width=\linewidth]{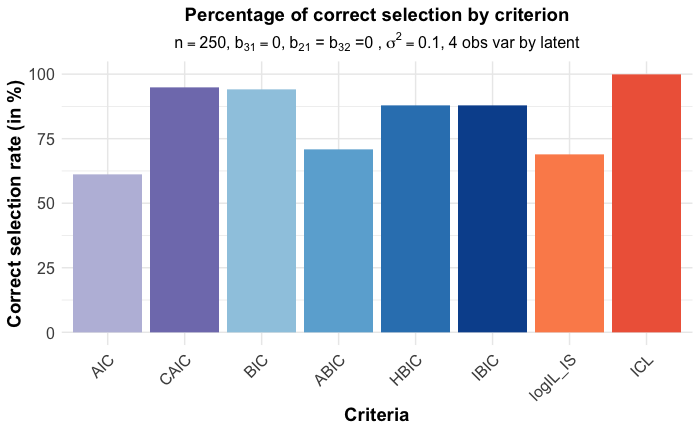}
        \caption{For $n=250$ and $b_{21}=b_{32}=0$.}
    \end{subfigure}

    \begin{subfigure}{0.45\textwidth}
        \centering
        \includegraphics[width=\linewidth]{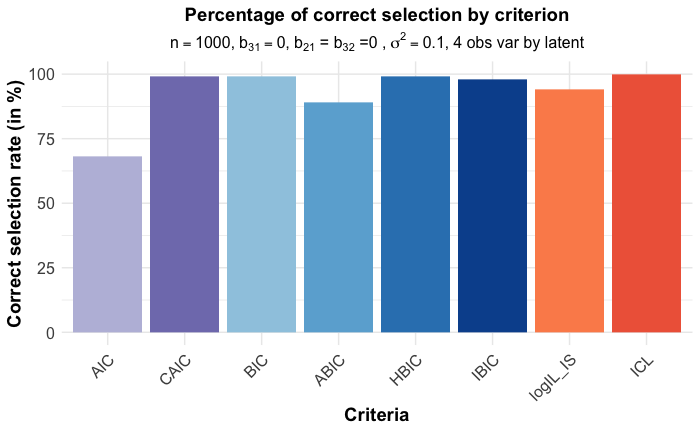}
        \caption{For $n=1000$ and $b_{21}=b_{32}=0$.}
    \end{subfigure}
    
    \caption{Results of correct selection when the true model is the null model.}
    \label{fig:res_null}
\end{figure}

Figure~\ref{fig:res_null} reports the results obtained when the true data-generating process corresponds to the null model. In this case, conservative criteria are favored, since the true model contains no structural effect. CAIC, BIC and ICL achieve the highest correct selection rates, reflecting their ability to avoid selecting unnecessary relationships.

By contrast, ABIC and logIL.IS tend to favor more complex models when the sample size is small, but their performance increases when $n$ is large enough, while the AIC performance stays low.

\paragraph{Overall comparison of the criteria.}
Across all simulation scenarios, model recovery is primarily driven by the sample size and the strength of the structural effects, a conclusion already discussed in \cite{lin_selecting_2017} regarding AIC and BIC-type criteria. Larger samples systematically improve correct selection rates, while stronger direct or indirect effects make the true model easier to identify. Nevertheless, the relative performance of the information criteria depends on the underlying model structure, confirming that no single criterion is uniformly optimal.

The classical criteria exhibit the expected behavior. AIC and ABIC are generally the best when the true model contains weak but nonzero effects, reflecting their greater sensitivity to small structural relationships. However, this flexibility also leads to a tendency to overselect complex models, particularly when the true model is null. Conversely, BIC and CAIC favor more parsimonious structures and perform remarkably well for the null and indirect models, but they may underselect complete models when the direct effect is weak.

The proposed logIL.IS criterion exhibits a behavior that is generally closer to AIC and ABIC than to BIC. It performs particularly well for identifying complete and direct models, where accounting for the latent structural information appears to improve the recovery of weak relationships. At the same time, logIL.IS outperforms the classical BIC when there is a direct effect but the relationships are weak, suggesting that incorporating the latent information structure into the penalty effectively reduces the tendency of BIC to underselect such models. Although logIL.IS is not systematically the best-performing criterion, it remains competitive across all scenarios and provides a good balance between sensitivity to weak effects and robustness across different latent structures. However, this gain comes at the price of increased computational cost. The importance sampling required to estimate the integrated likelihood is significantly more computationally expensive than the closed-form expressions underlying classical criteria such as AIC, BIC or ICL. This additional computational cost must therefore be taken into account when applying this criterion to large-scale problems. Note that the accuracy of the approximation depends on the number of samples $R$: increasing $R$ generally improves the precision of the estimated criterion, but also results in a proportional increase in computational time.

The proposed ICL criterion also shows promising performance under the present simulation settings. It is particularly effective for identifying the correct model when the direct effect is zero or large enough, where the integrated complete-data likelihood appears to provide useful information for discriminating between competing latent structures. In several configurations, ICL achieves the highest correct selection rates, especially when the sample size or the signal strength is sufficiently large. However, its performance deteriorates with weak structural direct effects, where it tends to favor simpler models. 

\subsubsection{Comparison of ICL and its oracle as residual variance increases}

Tables~\ref{tab:ICL70_03}--\ref{tab:ICL1000_03} summarize the model selection results obtained with the proposed ICL criterion when the measurement error variance is relatively large ($\sigma_j^2=0.3$). For each simulation configuration, the table reports the proportions (out of 100 simulated datasets) for which each candidate model (complete, indirect, direct, or null) is selected. The proportions in each row therefore sum to one. The shaded cells indicate the true data-generating model (so we want to maximize the value in these cells). Two versions of the criterion are compared. The first corresponds to the practical implementation, where the latent variables are replaced by their posterior estimates. The second is an oracle version in which the true latent variables are used when evaluating the integrated complete-data likelihood. Comparing these two versions allows us to isolate the impact of latent variable estimation on model selection.

For example, in the last row of Table~\ref{tab:ICL70_03}, where the true model is complete ($b_{21}=b_{32}=\sqrt{0.4}$ and $b_{31}=0.4$), the practical ICL correctly identifies the complete model in 59\% of the simulations, compared with 86\% for the oracle version. In the remaining cases, both criteria almost exclusively select the indirect model, indicating that the estimation error mainly affects the detection of the direct effect rather than the indirect structure.

\renewcommand{\arraystretch}{1.5}
\begin{table}[h]
\centering
\begin{tabular}{|c|c|cc|cc|cc|cc|}
\hline
& &
\multicolumn{2}{c|}{\textbf{Complete}} &
\multicolumn{2}{c|}{\textbf{Indirect}} &
\multicolumn{2}{c|}{\textbf{Direct}} &
\multicolumn{2}{c|}{\textbf{Null}}\\
\cline{3-10}
& & ICL & Oracle & ICL & Oracle & ICL & Oracle & ICL & Oracle \\
\hline

\multirow{3}{*}{$b_{21}=b_{32}=0$}
& $b_{31}=0$   & 0 & 0.01& 0 & 0.05 & 0.01 & 0.11 & \cellcolor{gray!20}0.99 & \cellcolor{gray!20}0.83\\
\hhline{|~|---------|}
& $b_{31}=0.1$ & 0 & 0 & 0 & 0.07 & \cellcolor{gray!20}0.06 & \cellcolor{gray!20}0.26& 0.94 & 0.67\\
\hhline{|~|---------|}
& $b_{31}=0.4$ & 0 & 0.06& 0 & 0.01 & \cellcolor{gray!20}0.65 & \cellcolor{gray!20}0.88& 0.35 & 0.05\\
\hline

\multirow{3}{*}{$b_{21}=b_{32}=\sqrt{0.1}$}
& $b_{31}=0$   & 0 & 0.08& \cellcolor{gray!20}0.35 & \cellcolor{gray!20}0.80& 0.03 & 0.02& 0.62 & 0.10\\
\hhline{|~|---------|}
& $b_{31}=0.1$ & \cellcolor{gray!20}0.02 & \cellcolor{gray!20}0.18& 0.36 & 0.71& 0.10 & 0.05& 0.52 & 0.06\\
\hhline{|~|---------|}
& $b_{31}=0.4$ & \cellcolor{gray!20}0.18 & \cellcolor{gray!20}0.79& 0.18 & 0.10 & 0.55 & 0.11 & 0.09 & 0 \\
\hline

\multirow{3}{*}{$b_{21}=b_{32}=\sqrt{0.4}$}
& $b_{31}=0$   & 0.03 & 0.11& \cellcolor{gray!20}0.97 & \cellcolor{gray!20}0.89& 0 & 0 & 0 & 0\\
\hhline{|~|---------|}
& $b_{31}=0.1$ & \cellcolor{gray!20}0.09 & \cellcolor{gray!20}0.27& 0.91 & 0.73 & 0 & 0& 0 & 0\\
\hhline{|~|---------|}
& $b_{31}=0.4$ & \cellcolor{gray!20}0.59 & \cellcolor{gray!20}0.86& 0.41 & 0.14 & 0 & 0 & 0 & 0\\
\hline

\end{tabular}
\caption{ICL and oracle results for $n=70$ and $\sigma_j^2=0.3$. The shaded cells indicate the true data-generating model that we want to select.}
\label{tab:ICL70_03}
\end{table}

\begin{table}[h]
\centering
\begin{tabular}{|c|c|cc|cc|cc|cc|}
\hline
& &
\multicolumn{2}{c|}{\textbf{Complete}} &
\multicolumn{2}{c|}{\textbf{Indirect}} &
\multicolumn{2}{c|}{\textbf{Direct}} &
\multicolumn{2}{c|}{\textbf{Null}}\\
\cline{3-10}
& & ICL & Oracle & ICL & Oracle & ICL & Oracle & ICL & Oracle\\
\hline

\multirow{3}{*}{$b_{21}=b_{32}=0$}
& $b_{31}=0$   & 0 & 0.01 & 0 & 0& 0 & 0.06 & \cellcolor{gray!20}1 & \cellcolor{gray!20}0.93\\
\hhline{|~|---------|}
& $b_{31}=0.1$ & 0& 0 & 0& 0.02 & \cellcolor{gray!20}0 & \cellcolor{gray!20}0.34& 1 & 0.64\\
\hhline{|~|---------|}
& $b_{31}=0.4$ & 0& 0.01 & 0& 0 & \cellcolor{gray!20}0.21 & \cellcolor{gray!20}0.99 & 0.79 & 0\\
\hline

\multirow{3}{*}{$b_{21}=b_{32}=\sqrt{0.1}$}
& $b_{31}=0$   & 0 & 0.09 & \cellcolor{gray!20}0.05 & \cellcolor{gray!20}0.91& 0 & 0 & 0.95 & 0\\
\hhline{|~|---------|}
& $b_{31}=0.1$ & \cellcolor{gray!20}0& \cellcolor{gray!20}0.30& 0.05 & 0.70 & 0 & 0& 0.95 & 0\\
\hhline{|~|---------|}
& $b_{31}=0.4$ & \cellcolor{gray!20}0.03 & \cellcolor{gray!20}1& 0.02 & 0& 0.68& 0 & 0.27 & 0\\
\hline

\multirow{3}{*}{$b_{21}=b_{32}=\sqrt{0.4}$}
& $b_{31}=0$   & 0 & 0.08 & \cellcolor{gray!20}1 & \cellcolor{gray!20}0.92 & 0 & 0& 0 & 0\\
\hhline{|~|---------|}
& $b_{31}=0.1$ & \cellcolor{gray!20}0.02 & \cellcolor{gray!20}0.29& 0.98 & 0.71 & 0 & 0& 0 & 0\\
\hhline{|~|---------|}
& $b_{31}=0.4$ & \cellcolor{gray!20}0.94 & \cellcolor{gray!20}0.99& 0.06 & 0.01& 0 & 0 &0 & 0 \\
\hline

\end{tabular}
\caption{ICL and oracle results for $n=250$ and $\sigma_j^2=0.3$. The shaded cells indicate the true data-generating model that we want to select.}
\label{tab:ICL250_03}
\end{table}

\begin{table}[h]
\centering
\begin{tabular}{|c|c|cc|cc|cc|cc|}
\hline
& &
\multicolumn{2}{c|}{\textbf{Complete}} &
\multicolumn{2}{c|}{\textbf{Indirect}} &
\multicolumn{2}{c|}{\textbf{Direct}} &
\multicolumn{2}{c|}{\textbf{Null}}\\
\cline{3-10}
& & ICL & Oracle & ICL & Oracle & ICL & Oracle & ICL & Oracle\\
\hline

\multirow{3}{*}{$b_{21}=b_{32}=0$}
& $b_{31}=0$   & 0 & 0 & 0 & 0& 0 & 0.01& \cellcolor{gray!20}1 & \cellcolor{gray!20}0.99\\
\hhline{|~|---------|}
& $b_{31}=0.1$ & 0 & 0 & 0& 0& \cellcolor{gray!20}0 & \cellcolor{gray!20}0.75& 1 & 0.25\\
\hhline{|~|---------|}
& $b_{31}=0.4$ & 0 & 0& 0& 0 & \cellcolor{gray!20}0 & \cellcolor{gray!20}1& 1 & 0\\
\hline

\multirow{3}{*}{$b_{21}=b_{32}=\sqrt{0.1}$}
& $b_{31}=0$   & 0 & 0.01 & \cellcolor{gray!20}0 & \cellcolor{gray!20}0.99& 0 & 0& 1 & 0\\
\hhline{|~|---------|}
& $b_{31}=0.1$ & \cellcolor{gray!20}0 & \cellcolor{gray!20}0.65& 0 & 0.35 & 0 & 0& 1 & 0\\
\hhline{|~|---------|}
& $b_{31}=0.4$ & \cellcolor{gray!20}0 & \cellcolor{gray!20}1& 0& 0 & 0 & 0 & 1 & 0\\
\hline

\multirow{3}{*}{$b_{21}=b_{32}=\sqrt{0.4}$}
& $b_{31}=0$   & 0 & 0.01& \cellcolor{gray!20}0 & \cellcolor{gray!20}0.99& 0 & 0& 1 & 0\\
\hhline{|~|---------|}
& $b_{31}=0.1$ & \cellcolor{gray!20}0& \cellcolor{gray!20}0.56 & 0 & 0.44 & 0 & 0& 1 & 0\\
\hhline{|~|---------|}
& $b_{31}=0.4$ & \cellcolor{gray!20}0 & \cellcolor{gray!20}1& 0 & 0& 0.16 & 0&0.84  & 0\\
\hline

\end{tabular}
\caption{ICL and oracle results for $n=1000$ and $\sigma_j^2=0.3$. The shaded cells indicate the true data-generating model that we want to select.}
\label{tab:ICL1000_03}
\end{table}

When the latent variables are estimated, the proposed ICL criterion performs poorly even for $n=1000$, indicating the problem is not finite‑sample variability. For the null model, it systematically selects the correct structure. However, for the direct, indirect and complete models, the criterion exhibits a strong tendency to select overly simple models. In particular, indirect models are frequently classified as null models, while complete models are often confused with direct or null structures depending on the signal strength. Since this problem is observed when the sample size is large, it suggests that these poor results cannot be explained solely by the variability associated with a finite sample.

The results are very different when we use the true latent variables rather than their estimates. In this oracle setting, the criterion correctly identifies all the model structures with high accuracy as soon as the direct effect becomes sufficiently strong. For $n=1000$, correct selection rates reach nearly perfect recovery for several configurations. These results indicate that the ICL criterion itself is capable of discriminating between the competing structural models when the latent variables are known.

The comparison between the two versions of the criterion strongly suggests that the poor empirical performance of ICL in this condition is mainly due to the estimation of the latent variables rather than to the criterion itself. Indeed, replacing the estimated latent variables by their true values almost completely removes the systematic underselection observed with the practical implementation. This finding highlights the impact of the uncertainty associated with latent variable estimation on the ICL criterion. Unlike mixture models, where the latent allocation variables are discrete and can often be estimated with high confidence, latent variables in SEM are continuous and estimated with non-negligible uncertainty, which propagates directly into the criterion.

Overall, these experiments show that the proposed ICL is effective provided there is sufficient information to correctly estimate the latent variables.

\section{Conclusion}
\label{sec_conclusion}

This paper revisited the problem of model selection in Gaussian structural equation models from the perspective of complete-data integrated likelihood. Unlike most previous studies, we considered both scenarios with and without a signal, which made it possible to evaluate not only the ability of information criteria to capture existing latent relationships but also their ability to avoid selecting superfluous relationships. This broader evaluation provides a more comprehensive analysis of model selection performance in structural equation modeling.

Two complementary criteria exploiting the latent structure were investigated. The first, logIL.IS, incorporates information on the latent structural dependencies through an importance sampling approximation of the integrated observed-data likelihood. The second adapts the Integrated Completed Likelihood (ICL) criterion to Gaussian SEM by directly exploiting the complete-data likelihood.

The simulation study shows that explicitly accounting for the latent structure can indeed improve model selection, although the benefits depend on the way this information is incorporated. logIL.IS consistently provides competitive performance across the different model configurations considered. It generally improves on the classical BIC when latent dependency structures become more complex, while remaining robust in simpler situations where avoiding overfitting is essential. These results suggest that adapting the complexity penalty to the latent dependency structure is an effective strategy to improve model selection in SEM.

The proposed ICL criterion follows a different philosophy by relying directly on the integrated complete-data likelihood. The numerical experiments indicate that, while more sensitive to the intensity of the direct signal, ICL is particularly effective for recovering complex latent structures when the latent variables are estimated with sufficient accuracy. However, its performance deteriorates when the measurement model provides only limited information about latent variables, highlighting the crucial role that latent variable estimation plays in the practical implementation of complete-data criteria. This finding suggests that the main limitation lies less in the complete-data formulation itself than in the uncertainty associated with estimating the latent variables.

Overall, these results demonstrate the potential of information criteria specifically designed for latent variable models, rather than those directly adapted to contexts where all variables are observed. They also highlight that no single criterion is universally optimal in all situations. Rather, the choice of criterion must take into account the size of the available sample, the expected complexity of the latent structure, and the quality of its measurement.

Several directions for future work naturally emerge from this study. On the methodological side, improving complete-data criteria by explicitly accounting for the uncertainty of latent variable estimation appears particularly promising. More generally, extending the proposed approaches to non-Gaussian latent variables or misspecified measurement models would provide a broader assessment of their practical usefulness. A final direction concerns the specification of the measurement model, and in particular the choice of the number of observed variables associated with each latent variable. While this choice is often based on domain-specific knowledge, it also has significant statistical implications, as it directly affects the identifiability, estimation accuracy, and interpretability of the latent variables. Developing data-driven procedures to select the number of observed variables while preserving the validity of the resulting estimates therefore constitutes a major methodological challenge.

\section*{Acknowledgements}
This work received government funding managed by the Agence Nationale de la Recherche under the France 2030 program as part of the Agroecology and Digital research program, reference number ANR-22-PEAE-0015.

\bibliography{bibliography}
\bibliographystyle{apalike}

\newpage
\begin{appendices}

\section{Proof of Proposition \ref{propICL}}
\label{app_proof}

\begin{proof}[Proof of Proposition \ref{propICL}]
    Note that
\begin{align*}
    p(X,Z \mid \vartheta,\bm{f}, \bm{\omega})
    =& \int p(X,Z \mid \theta; \bm{f}, \bm{\omega}) p(\theta \mid \vartheta;\bm{f}, \bm{\omega}) d\theta \\
=&\int p(X \mid Z, \Lambda, \Sigma) p(\Lambda, \Sigma \mid \vartheta, \bm{\omega}) d\Lambda d\Sigma \times \int p(Z \mid B,\bm{f}) p(B \mid \vartheta,\bm{f}) dB \\
=& J_1(X \mid Z, \bm{\omega}) J_2(Z \mid \bm{f})
\end{align*} 
where $J_1(X \mid Z, \bm{\omega}) = \int p(X \mid Z, \Lambda, \Sigma) p(\Lambda, \Sigma \mid \vartheta, \bm{\omega}) d\Lambda d\Sigma$ and $J_2(Z \mid \bm{f}) = \int p(Z \mid B,\bm{f}) p(B \mid \vartheta,\bm{f}) dB$.

\begin{align*}
    J_1(X \mid Z, \bm{\omega}) =& \int p(X \mid Z, \Lambda, \Sigma) p(\Lambda, \Sigma \mid \vartheta, \bm{\omega}) d\Lambda d\Sigma \\
    =& \int \prod_{i=1}^n
\phi_p(X_i; \Lambda Z_i, \Sigma) \times \prod_{j=1}^p
\Bigg(
g(\sigma_j^2; \alpha_j/2, \beta_j^2/2)
\prod_{h=1}^q
\left[
\phi_1(\Lambda_{jh}; \nu_j, \sigma_j^2 \delta_j^{-1})
\right]^{\omega_{jh}}
\Bigg) d\Lambda d\Sigma \\
=& \int \prod_{j=1}^p
\Bigg(
g(\sigma_j^2; \alpha_j/2, \beta_j^2/2)
\prod_{h=1}^q
\left[
\phi_1(\Lambda_{jh}; \nu_j, \sigma_j^2 \delta_j^{-1}) \prod_{i=1}^n
\phi_1(x_{ij}; \Lambda_{jh} z_{ih}, \sigma_j^2 )
\right]^{\omega_{jh}}
\Bigg) d\Lambda d\Sigma \\
=& \prod_{j=1}^p \int
\Bigg(
g(\sigma_j^2; \alpha_j/2, \beta_j^2/2)\phi_1(\lambda_{j}; \nu_j, \sigma_j^2 \delta_j^{-1})
 \prod_{i=1}^n
\phi_1(x_{ij}; \lambda_j \sum_{h=1}^q \omega_{jh} z_{ih}, \sigma_j^2 ) 
\Bigg) dv_j,
\end{align*}
where $v_j=(\sigma_j^2, \Lambda_{j1}, \dots, \Lambda_{jq})$ and $\lambda_j$ is the unique non-zero element in row $j$ of $\Lambda$.

To calculate each of the $p$ integrals of $J_1$, we need the following lemma.

\begin{lemma}
\label{lemma}
Let $v = (v_1^\top, \dots, v_n^\top)^\top$ be $n$ independent realizations, 
where $v_i = (v_{i1}, \dots, v_{id})^\top \in \mathbb{R}^d$, generated 
given $u = (u_1^\top, \dots, u_n^\top)$, where 
$u_i = (u_{i1}, \dots, u_{ir})^\top \in \mathbb{R}^r$ by a Gaussian regression model such that the components of vector $v_i$ are conditionally independent given $u_i$ and such that
\[
v_{ij} \mid u_i \sim \mathcal{N}(\lambda_j^\top u_i, \sigma_j^2).
\]

Consider that the prior on $\sigma^2 = (\sigma_1^2, \dots, \sigma_d^2)^\top$ is a product of univariate priors such that
\[
\sigma_j^2 \sim \text{IG}(\alpha_j/2, \beta_j^2/2).
\]

In addition, assume that the prior on $\lambda = (\lambda_1^\top, \dots, \lambda_d^\top)^\top$ given $\sigma$ is a product prior distribution on $\lambda_j$ given $\sigma_j^2$, where
\[
\lambda_j \mid \sigma_j^2 \sim \mathcal{N}_r(\nu_j, \sigma_j^2 \Delta_j^{-1}),
\]
with $\Delta_j = \operatorname{diag}(\delta_1, \dots, \delta_r)$.

Then,
\[
I_n(v \mid u, \vartheta)
=
\int p(v, \sigma^2, \lambda \mid u, \vartheta)\, d\sigma^2 \, d\lambda
\]
is given by
\[
I_n(v \mid u, \vartheta)
=
\prod_{j=1}^d
\frac{\det^{1/2}(\Delta_j)}
{\det^{1/2}(S_j)}
\frac{1}{\pi^{n/2}}
\frac{\Gamma\left(n/2 + \alpha_j/2\right)}
{\Gamma\left(\frac{\alpha_j}{2}\right)}
\frac{\beta_j^{\alpha_j}}
{\left(\beta_j^2 + t_j\right)^{\alpha_j/2 + n/2}},
\]
where
\[t_j=\sum_{i=1}^n v_{ij}^2+\nu_j^\top \Delta_j \nu_j-m_j^\top S_j m_j, \quad S_j=\Delta_j+\sum_{i=1}^n u_i u_i^\top, \quad m_j=S_j^{-1}\left(\sum_{i=1}^n u_i v_{ij}+\Delta_j \nu_j\right).\]
\end{lemma}

Then we apply Lemma \ref{lemma} to each of the $p$-integrals that define $J_1$ with $v=X_{.j} \in \mathbb{R}^n$, $u_i=\omega_j^\top Z_i = \sum_{h=1}^q \omega_{jh}z_{ih}$, that is $u = \sum_{h=1}^q \omega_{jh}Z_{.h} \in \mathbb{R}^n$, $d=1$, $r=1$, $\Delta_j=\delta_j$. We have $x_{ij} \mid Z_i \sim \mathcal{N}(\lambda_j^\top \sum_{h=1}^q \omega_{jh}z_{ih}, \sigma_j^2)$, $\sigma_j^2 \sim \text{IG}(\alpha_j/2, \beta_j^2/2)$, $\lambda_j \mid \sigma_j^2 \sim \mathcal{N}_r(\nu_j, \sigma_j^2 \delta_j^{-1})$. So $S_j=\Delta_j+\sum_{i=1}^n u_i u_i^\top=\delta_j+ \omega_j^\top \sum_{i=1}^n Z_i Z_i^\top \omega_j = s_j$, $m_j=S_j^{-1}\left(\sum_{i=1}^n u_i v_{ij}+\Delta_j \nu_j\right) = s_j^{-1} (\sum_{i=1}^n \bm{\omega}_j^\top Z_i x_{ij} + \delta_j \nu_j) = a_j$, and $t_j=\sum_{i=1}^n v_{ij}^2+\nu_j^\top \Delta_j \nu_j-m_j^\top S_j m_j = t_j=\sum_{i=1}^n x_{ij}^2+\delta_j \nu_j^2-s_j a_j^2=t_j$. Therefore,

$J_1(X, Z) = \prod_{j=1}^p \frac{\delta_j^{1/2}}{\pi^{n/2} s_j^{1/2}} \frac{\Gamma(n/2+\alpha_j/2)}{\Gamma(\alpha_j/2)}
\frac{\beta_j^{\alpha_j}}
{(\beta_j^2+t_j)^{\alpha_j/2+n/2}} = \dfrac{1}{\pi^{np/2}} \prod_{j=1}^p \dfrac{\delta_j^{1/2}}{s_j^{1/2}} \dfrac{\Gamma(n/2+\alpha_j/2)}{\Gamma(\alpha_j/2)}
\dfrac{\beta_j^{\alpha_j}}
{(\beta_j^2+t_j)^{\alpha_j/2+n/2}}$.

Now we have to calculate $J_2(Z \mid \bm{f}) = \int p(Z \mid B,\bm{f}) p(B \mid \vartheta,\bm{f}) dB$. We have 
\[ p(Z \mid B,\bm{f}) = \prod_{i=1}^n p(Z_i \mid B,\bm{f}), \quad p(Z_i \mid B,\bm{f}) \sim \mathcal{N}_q(0, \Psi), \]
where $\Psi = (I_q - B)^{-1} \left[(I_q - B)^{-1}\right]^\top$ because $\Gamma=Id$. 
\[ p(Z \mid B,\bm{f}) = \dfrac{1}{(2\pi)^{nq/2} \det^{n/2}(\Psi)} \exp \left\{ - \frac{1}{2} \sum_{i=1}^n Z_i^\top \Psi^{-1} Z_i \right\}  .\]
We have $\det(\Psi)=\det^{-2}(I_q-B)=1$ since $B$ is a strict triangular matrix, therefore with zeros on the diagonal. Now for the computation of the sum

\begin{align*}
    \sum_{i=1}^n Z_i^\top \Psi^{-1} Z_i =& \sum_{i=1}^n Z_i^\top (I_q - B)^\top (I_q - B) Z_i = \sum_{i=1}^n [(I_q - B)Z_i]^\top (I_q - B) Z_i \\
    =& \sum_{i=1}^n \sum_{h=1}^q ((I_q - B)Z_i)_h^2 \\
    =& \sum_{i=1}^n \sum_{h=1}^q \left(z_{ih}-\sum_{\{\ell : f_{h\ell}=1\}} B_{h\ell} z_{i\ell} \right)^2\\
    =& \sum_{i=1}^n \sum_{h=1}^q \left( z_{ih}^2 - 2 \sum_{\{\ell : f_{h\ell}=1\}} B_{h\ell} z_{i\ell} z_{ih} + \sum_{\{\ell : f_{h\ell}=1\}} \sum_{\{k : f_{hk}=1\}} B_{h\ell} z_{i\ell} z_{ik} B_{hk} \right) \\
    =& \sum_{h=1}^q \left( \sum_{i=1}^n z_{ih}^2 - 2 \bm{\tilde{m}}_h^\top \bm{m}_h + \bm{m}_h^\top \tilde{M}_h \bm{m}_h \right)
\end{align*}
where, as a reminder, $\bm{m}_h$ is the $d_h$-dimensional vector composed of the non-zero elements in row $h$ of $B$, $\tilde M_h$ is the symmetric matrix of dimension $d_h \times d_h$ composed of rows and columns of matrix $\sum_{i=1}^n Z_i Z_i^\top$ having a index $\ell$ such that $f_{h\ell}=1$ and $\bm{\tilde m}_h$ is the $d_h$-dimensional vector composed of the elements $\ell$ of
row $h$ of $\sum_{i=1}^n Z_i z_{ih}$ such that $f_{h\ell}=1$. So 
\[\sum_{i=1}^n Z_i^\top \Psi^{-1} Z_i = \sum_{h=1}^q \left[\sum_{i=1}^n z_{ih}^2 - \tilde{\bm{m}}_h^\top \tilde{M}_h^{-1} \tilde{\bm{m}}_h \right]  + \sum_{h=1}^q (\bm{m}_h - \tilde{M}_h^{-1} \tilde{\bm{m}}_h)^\top \tilde{M}_h (\bm{m}_h - \tilde{M}_h^{-1} \tilde{\bm{m}}_h). \] Since we are working conditionally on $\bm{f}$, the first term does not depend on $B$ and can therefore be taken out of the integral.

We have also 
\[p(B \mid \vartheta,\bm{f}) = \prod_{h=1}^q \phi_{d_h}(\bm{m}_h; \bm{\mu}_h, \bm{\kappa}_h^{-1}) = \prod_{h=1}^q \dfrac{1}{(2\pi)^{d_h/2} \det^{1/2}(\bm{\kappa}_h^{-1})} \exp \left\{ -\frac{1}{2}(\bm{m}_h-\bm{\mu}_h)^\top \bm{\kappa}_h (\bm{m}_h-\bm{\mu}_h)\right\}.\]

Therefore 
\begin{align*}
    &J_2(Z \mid \bm{f}) = \int p(Z \mid B,\bm{f}) p(B \mid \vartheta,\bm{f}) dB \\
    =& \int \dfrac{1}{(2\pi)^{nq/2}} \exp \left\{ - \frac{1}{2} \sum_{i=1}^n Z_i^\top \Psi^{-1} Z_i \right\} \prod_{h=1}^q \dfrac{\det^{1/2}(\bm{\kappa}_h)}{(2\pi)^{d_h/2}} \exp \left\{ -\frac{1}{2}(\bm{m}_h-\bm{\mu}_h)^\top \bm{\kappa}_h (\bm{m}_h-\bm{\mu}_h)\right\} dB\\
    =& \dfrac{1}{(2\pi)^{nq/2}} \exp \left\{-\frac{1}{2}\left(\sum_{h=1}^q \sum_{i=1}^n z_{ih}^2 - \sum_{h=1}^q\tilde{\bm{m}}_h^\top \tilde{M}_h^{-1} \tilde{\bm{m}}_h \right)\right\} \times \\
    &\prod_{h=1}^q \left(\dfrac{\det^{1/2}(\bm{\kappa}_h)}{(2\pi)^{d_h/2}} \int \exp \left\{ -\frac{1}{2} \left((\bm{m}_h - \tilde{M}_h^{-1} \tilde{\bm{m}}_h)^\top \tilde{M}_h (\bm{m}_h - \tilde{M}_h^{-1} \tilde{\bm{m}}_h) + (\bm{m}_h-\bm{\mu}_h)^\top \bm{\kappa}_h (\bm{m}_h-\bm{\mu}_h) \right)\right\} d\bm{m}_h \right)^{\mathds{1}_{\{d_h>0\}}}
\end{align*}
since when $d_h=0$, $\bm{m}_h=\emptyset$ and therefore everything inside the product is equal to $1$.

We have
\footnotesize{
\begin{equation*}
    (\bm{m}_h - \tilde{M}_h^{-1} \tilde{\bm{m}}_h)^\top \tilde{M}_h (\bm{m}_h - \tilde{M}_h^{-1} \tilde{\bm{m}}_h) + (\bm{m}_h-\bm{\mu}_h)^\top \bm{\kappa}_h (\bm{m}_h-\bm{\mu}_h) = (\bm{m}_h - \bm{\zeta}_h)^\top \tilde S_h(\bm{m}_h - \bm{\zeta}_h) - \bm{\zeta}_h^\top \tilde S_h \bm{\zeta}_h + \tilde{\bm{m}}_h^\top \tilde{M}_h^{-1} \tilde{\bm{m}}_h + \bm{\mu}_h^\top \bm{\kappa}_h \bm{\mu}_h
\end{equation*}}
\normalsize
with $\tilde S_h = \tilde M_h + \bm{\kappa}_h$, $\bm{\zeta}_h = \tilde S_h^{-1}(\bm{\tilde m}_h +\bm{\kappa}_h \bm{\mu}_h)$. So, 

\begin{align*}
    \int &\exp \left\{ -\frac{1}{2} \left((\bm{m}_h - \tilde{M}_h^{-1} \tilde{\bm{m}}_h)^\top \tilde{M}_h (\bm{m}_h - \tilde{M}_h^{-1} \tilde{\bm{m}}_h) + (\bm{m}_h-\bm{\mu}_h)^\top \bm{\kappa}_h (\bm{m}_h-\bm{\mu}_h) \right)\right\} d\bm{m}_h \\
    =& \int \exp\left\{ - \frac{1}{2} (\bm{m}_h - \bm{\zeta}_h)^\top \tilde S_h(\bm{m}_h - \bm{\zeta}_h) \right\} d\bm{m}_h \exp \left\{ - \frac{1}{2} (-\bm{\zeta}_h^\top \tilde S_h \bm{\zeta}_h + \tilde{\bm{m}}_h^\top \tilde{M}_h^{-1} \tilde{\bm{m}}_h + \bm{\mu}_h^\top \bm{\kappa}_h \bm{\mu}_h)\right\} \\
    =& \text{ det}^{-1/2}(\tilde S_h) (2\pi)^{d_h/2} \exp \left\{ - \frac{1}{2} (-\bm{\zeta}_h^\top \tilde S_h \bm{\zeta}_h + \tilde{\bm{m}}_h^\top \tilde{M}_h^{-1} \tilde{\bm{m}}_h + \bm{\mu}_h^\top \bm{\kappa}_h \bm{\mu}_h)\right\}
\end{align*}
by the normalization constant of the Gaussian law $\mathcal{N}_{d_h}(\bm{\zeta}_h,\tilde S_h^{-1})$.

Therefore,
\small{
\begin{align*}
    J_2(Z \mid \bm{f}) = &\dfrac{1}{(2\pi)^{nq/2}} \exp \left\{-\frac{1}{2}\left(\sum_{h=1}^q \sum_{i=1}^n z_{ih}^2 - \sum_{h=1}^q\tilde{\bm{m}}_h^\top \tilde{M}_h^{-1} \tilde{\bm{m}}_h \right)\right\} \times \\
    & \prod_{h=1}^q \left(\dfrac{\det^{1/2}(\bm{\kappa}_h)}{\det^{1/2}(\tilde S_h)} \exp \left\{ - \frac{1}{2} (-\bm{\zeta}_h^\top \tilde S_h \bm{\zeta}_h + \tilde{\bm{m}}_h^\top \tilde{M}_h^{-1} \tilde{\bm{m}}_h + \bm{\mu}_h^\top \bm{\kappa}_h \bm{\mu}_h)\right\} \right)^{\mathds{1}_{\{d_h>0\}}} \\
    =& \dfrac{1}{(2\pi)^{nq/2}} \prod_{h=1}^q \exp \left\{ -\frac{1}{2} \sum_{i=1}^n z_{ih}^2 \right\}\left(\dfrac{\det^{1/2}(\bm{\kappa}_h)}{\det^{1/2}(\tilde S_h)} \exp \left\{ - \frac{1}{2} (- \tilde{\bm{m}}_h^\top \tilde{M}_h^{-1} \tilde{\bm{m}}_h-\bm{\zeta}_h^\top \tilde S_h \bm{\zeta}_h + \tilde{\bm{m}}_h^\top \tilde{M}_h^{-1} \tilde{\bm{m}}_h + \bm{\mu}_h^\top \bm{\kappa}_h \bm{\mu}_h)\right\} \right)^{\mathds{1}_{\{d_h>0\}}} \\
    =& \dfrac{1}{(2\pi)^{nq/2}} \prod_{h=1}^q \exp \left\{ -\frac{1}{2} \sum_{i=1}^n z_{ih}^2 \right\} \left[\frac{\det^{1/2}(\bm{\kappa}_h)}{\det^{1/2}(\tilde S_h)}\exp\left(-\frac{\tilde t_h}{2}\right)\right]^{\mathds{1}_{\{d_h>0\}}}
\end{align*}}
\normalsize
where $\tilde t_h =  \bm{\mu}_h^\top \bm{\kappa}_h \bm{\mu}_h - \bm{\zeta}_h^\top \tilde S_h \bm{\zeta}_h$.

\color{black}
We conclude that
\begin{align*}
    p(X,Z \mid \vartheta,\bm{f}, \bm{\omega})
=
&\frac{1}{\pi^{(p+q)n/2} 2^{qn/2}}
\prod_{h=1}^q \exp \left\{ -\frac{1}{2} \sum_{i=1}^n z_{ih}^2 \right\}
\left[
\frac{\det^{1/2}(\bm{\kappa}_h)}{\det^{1/2}(\tilde S_h)}
\exp\left(-\frac{\tilde t_h}{2}\right)
\right]^{\mathds{1}_{\{d_h>0\}}} \times \\
&\prod_{j=1}^p
\left[
\frac{\delta_j^{1/2}}{s_j^{1/2}}
\frac{\Gamma(n/2+\alpha_j/2)}{\Gamma(\alpha_j/2)}
\frac{\beta_j^{\alpha_j}}
{(\beta_j^2+t_j)^{n/2+\alpha_j/2}}
\right].
\end{align*}

\end{proof}

\begin{proof}[Proof of Lemma \ref{lemma}]

Let $v = (v_1^\top, \dots, v_n^\top)^\top$, 
where $v_i = (v_{i1}, \dots, v_{id})^\top \in \mathbb{R}^d$, generated 
given $u = (u_1^\top, \dots, u_n^\top)$, where 
$u_i = (u_{i1}, \dots, u_{ir})^\top \in \mathbb{R}^r$ by 
\[
v_{ij} \mid u_i \overset{\text{ind}}{\sim} \mathcal{N}(\lambda_j^\top u_i, \sigma_j^2).
\]
In addition, the priors on $\sigma^2 = (\sigma_1^2, \dots, \sigma_d^2)^\top$ and $\lambda = (\lambda_1^\top, \dots, \lambda_d^\top)^\top$ given $\sigma$ are
\[
\sigma_j^2 \sim \text{IG}(\alpha_j/2, \beta_j^2/2).
\]
\[
\lambda_j \mid \sigma_j^2 \sim \mathcal{N}_r(\nu_j, \sigma_j^2 \Delta_j^{-1}),
\]
with $\Delta_j^{-1} = \operatorname{diag}(\delta_1, \dots, \delta_r)$.

We want to compute 
\[
I_n(v \mid u, \vartheta)
=
\int p(v, \sigma^2, \lambda \mid u, \vartheta)\, d\sigma^2 \, d\lambda,
\]
where $\vartheta$ groups all the hyper-parameters.

We have 
\begin{align*}
    p(v, \sigma^2, \lambda \mid u, \vartheta) =& p(v \mid u, \sigma^2, \lambda, \vartheta) p(\lambda \mid u, \sigma^2, \vartheta) p(\sigma^2 \mid u, \vartheta) \\
    =& \prod_{j=1}^{d} g\left(\sigma_j^2; \alpha_j/2, \beta_j^2/2\right) p(\mathbf{v}_{\bullet j} \mid u, \lambda_j , \sigma_j^2, \vartheta) p(\lambda_j \mid u, \sigma_j^2, \vartheta) \\
    =& \prod_{j=1}^{d} g\left(\sigma_j^2; \alpha_j/2, \beta_j^2/2\right) \dfrac{1}{(2\pi \sigma_j^2)^{n/2}} \exp \left\{-\frac{1}{2\sigma_j^2}\sum_{i=1}^n(v_{ij}-\lambda_j^\top u_i)^2 \right \} \times \\
    &\dfrac{\det^{1/2}(\Delta_j)}{(2\pi \sigma_j^2)^{r/2}} \exp \left\{-\frac{1}{2\sigma_j^2} (\lambda_j - \nu_j)^\top \Delta_j (\lambda_j - \nu_j) \right \} \\
    =& \prod_{j=1}^{d} g\left(\sigma_j^2; \alpha_j/2, \beta_j^2/2\right) \dfrac{\det^{1/2}(\Delta_j)}{(2\pi \sigma_j^2)^{(n+r)/2}} \exp \left\{-\frac{1}{2\sigma_j^2} \left(\sum_{i=1}^n (v_{ij}-\lambda_j^\top u_i)^2 + (\lambda_j - \nu_j)^\top \Delta_j (\lambda_j - \nu_j) \right)\right \}
\end{align*}

First, we analyze the terms within the exponential of the gaussian densities:

\begin{align*}
    \sum_{i=1}^n (v_{ij}-\lambda_j^\top u_i)^2 +& (\lambda_j - \nu_j)^\top \Delta_j (\lambda_j - \nu_j) \\
    =& \sum_{i=1}^n v_{ij}^2 -2 \sum_{i=1}^n v_{ij}\lambda_j^\top u_i + \sum_{i=1}^n \lambda_j^\top u_i u_i^\top \lambda_j + \lambda_j^\top \Delta_j \lambda_j - \lambda_j^\top \Delta_j \nu_j - \nu_j^\top \Delta_j \lambda_j + \nu_j^\top \Delta_j \nu_j \\
    =& \sum_{i=1}^n v_{ij}^2 + \nu_j^\top \Delta_j \nu_j -2 \lambda_j^\top \sum_{i=1}^n v_{ij} u_i + \lambda_j^\top \left( \sum_{i=1}^n u_i u_i^\top \right)\lambda_j + \lambda_j^\top \Delta_j \lambda_j - 2\lambda_j^\top \Delta_j \nu_j \\
\end{align*}
since $\lambda_j^\top \Delta_j \nu_j = \nu_j^\top \Delta_j \lambda_j $ since $\lambda_j^\top \Delta_j \nu_j\in \mathbb{R}$. Thus, with $S_j=\Delta_j+\sum_{i=1}^n u_i u_i^\top, m_j=S_j^{-1}\left(\sum_{i=1}^n u_i v_{ij}+\Delta_j \nu_j\right)$:

\[\sum_{i=1}^n (v_{ij}-\lambda_j^\top u_i)^2 + (\lambda_j - \nu_j)^\top \Delta_j (\lambda_j - \nu_j)= (\lambda_j-m_j)^\top S_j (\lambda_j-m_j) + \sum_{i=1}^n v_{ij}^2 + \nu_j^\top \Delta_j \nu_j - m_j^\top S_j m_j.\]
By substituting in the exponential, we recognize a law 
\[\phi_r(\lambda_j ; m_j, \sigma_j^2 S_j^{-1})=\dfrac{\det^{1/2}(S_j)}{(2\pi \sigma_j^2)^{r/2}}\exp \left\{ -\frac{1}{2\sigma_j^2} (\lambda_j-m_j)^\top S_j (\lambda_j-m_j)\right\}\]
So,
\small{
\[p(v, \sigma^2, \lambda \mid u, \vartheta) = \prod_{j=1}^{d} g\left(\sigma_j^2; \alpha_j/2, \beta_j^2/2\right) \frac{\det^{1/2}(\Delta_j)}{(2\pi \sigma_j^2)^{n/2}\det^{1/2}(S_j)} \exp \left\{-\frac{1}{2\sigma_j^2} \left(\sum_{i=1}^n v_{ij}^2 + \nu_j^\top \Delta_j \nu_j - m_j^\top S_j m_j\right)\right \} \phi_r(\lambda_j ; m_j, \sigma_j^2 S_j^{-1}).\]}

Let's note $t_j=\sum_{i=1}^n v_{ij}^2+\nu_j^\top \Delta_j \nu_j-m_j^\top S_j m_j$, and we want to integrate this expression with respect to $\lambda_j$ and $\sigma_j^2$. 
\begin{align*}
    I_n(v \mid u, \vartheta)&=
\int p(v, \sigma^2, \lambda \mid u, \vartheta)\, d\sigma^2 \, d\lambda \\
&= \prod_{j=1}^d \int g\left(\sigma_j^2; \alpha_j/2, \beta_j^2/2\right) \frac{\det^{1/2}(\Delta_j)}{(2\pi \sigma_j^2)^{n/2}\det^{1/2}(S_j)} \exp \left\{-\frac{t_j}{2\sigma_j^2}\right \} \underbrace{\int \phi_r(\lambda_j ; m_j, \sigma_j^2 S_j^{-1}) d\lambda_j}_{=1} d\sigma_j^2 \\
&= \prod_{j=1}^d \int g\left(\sigma_j^2; \alpha_j/2, \beta_j^2/2\right) \frac{\det^{1/2}(\Delta_j)}{(2\pi \sigma_j^2)^{n/2}\det^{1/2}(S_j)} \exp \left\{-\frac{t_j}{2\sigma_j^2}\right \} d\sigma_j^2
\end{align*}

We have the Inverse Gamma distribution
\[g\left(\sigma_j^2; \alpha_j/2, \beta_j^2/2\right) = \frac{1}{\Gamma(\alpha_j/2)}(\beta_j^2/2)^{\alpha_j/2} (1/\sigma_j^2)^{\alpha_j/2+1}\exp\left(-\frac{\beta_j^2}{2\sigma_j^2}\right)\]

So, we recognize a new inverse gamma distribution

\begin{align*}
    g\left(\sigma_j^2; \alpha_j/2, \beta_j^2/2\right) &\frac{\det^{1/2}(\Delta_j)}{(2\pi \sigma_j^2)^{n/2}\det^{1/2}(S_j)} \exp \left\{-\frac{t_j}{2\sigma_j^2}\right \} \\
    &= \frac{\det^{1/2}(\Delta_j)}{\Gamma(\alpha_j/2)\pi^{n/2}\det^{1/2}(S_j)} \frac{\beta_j^{\alpha_j}}{2^{n/2+\alpha_j/2}} \frac{1}{(\sigma_j^2)^{n/2+\alpha_j/2+1}} \exp \left\{-\frac{t_j+\beta_j^2}{2\sigma_j^2}\right \} \\
    &= \frac{\det^{1/2}(\Delta_j)}{\Gamma(\alpha_j/2)\pi^{n/2}\det^{1/2}(S_j)} g\left(\sigma_j^2; n/2+\alpha_j/2, (\beta_j^2+t_j)/2\right) \Gamma(n/2+\alpha_j/2) \frac{\beta_j^{\alpha_j}}{(\beta_j^2+t_j)^{n/2+\alpha_j/2}}
\end{align*}

Now by integrating with respect to $\sigma_j^2$ we conclude:

\[
I_n(v \mid u, \vartheta)
=
\prod_{j=1}^d
\frac{\det^{1/2}(\Delta_j)}
{\det^{1/2}(S_j)}
\frac{1}{\pi^{n/2}}
\frac{\Gamma\left(n/2 + \alpha_j/2\right)}
{\Gamma\left(\alpha_j/2\right)}
\frac{\beta_j^{\alpha_j}}
{\left(\beta_j^2 + t_j\right)^{\alpha_j/2 + n/2}}.
\]
    
\end{proof}

\end{appendices}

\end{document}